\documentclass[11pt]{article}
\usepackage{mathtools}

\usepackage[margin=1in]{geometry}
\usepackage{sn-preamble}
\usepackage{latexsym}
\usepackage{amsmath}
\usepackage{amssymb}
\usepackage{amsthm}
\usepackage{qcircuit}
\usepackage{enumitem}
\usetikzlibrary{shapes.geometric}
\usetikzlibrary{patterns}
\usepackage{microtype}

\newcommand{\ignore}[1]{}

\newcommand{\Pptf}{\calC_{\mathrm{PTF}}}

\newcommand{\dtv}{\mathrm{d}_{\mathrm{TV}}}
\newcommand{\tr}{\mathrm{tr}}

\newcommand{\PTF}{\textsc{PTF-Distinguisher}}

\newcommand{\tx}[1]{\wt{\bx}^{(#1)}}
\newcommand{\ty}[1]{\wt{\by}^{(#1)}}

\newcommand{\mun}{\mu^{\otimes n}}

\newcommand{\vcdim}{\mathrm{VC\text{-}Dim}}

\title{
Detecting Low-Degree Truncation
\vspace{0.5em}
}

\author{
Anindya De\thanks{University of Pennsylvania. Email: \href{mailto:anindyad@cis.upenn.edu}{\texttt{anindyad@cis.upenn.edu}}.}
\and 
Huan Li\thanks{University of Pennsylvania. Email: \href{mailto:huanli@cis.upenn.edu}{\texttt{huanli@cis.upenn.edu}}.}
\and 
Shivam Nadimpalli\thanks{Columbia University. Email: \href{mailto:sn2855@columbia.edu}{\texttt{sn2855@columbia.edu}}.}
\and 
Rocco A. Servedio\thanks{Columbia University. Email: \href{mailto:ras2105@columbia.edu}{\texttt{ras2105@columbia.edu}}.}
\vspace{0.5em}
}

\date{
\today
}

\begin{document}

\pagenumbering{gobble}
\maketitle

\begin{abstract}
We consider the following basic, and very broad, statistical problem:  Given a known high-dimensional distribution ${\cal D}$ over $\R^n$ and a collection of data points in $\R^n$, distinguish between the two possibilities that (i) the data was drawn from ${\cal D}$, versus (ii) the data was drawn from ${\cal D}|_S$, i.e. from ${\cal D}$ subject to truncation by an unknown truncation set $S \subseteq \R^n$.

We study this problem in the setting where ${\cal D}$ is a high-dimensional i.i.d.~product distribution and  $S$ is an unknown degree-$d$ \emph{polynomial threshold function} (one of the most well-studied types of Boolean-valued function over $\R^n$).  Our main results are an efficient algorithm when ${\cal D}$ is a \emph{hypercontractive} distribution, and a matching lower bound:

\begin{itemize}

\item For any constant $d$, we give a polynomial-time algorithm which successfully distinguishes ${\cal D}$ from ${\cal D}|_S$ using $O(n^{d/2})$ samples (subject to mild technical conditions on ${\cal D}$ and $S$);

\item Even for the simplest case of ${\cal D}$ being the uniform distribution over $\bits^n$, we show that for any constant $d$, any distinguishing algorithm for degree-$d$ polynomial threshold functions must use $\Omega(n^{d/2})$ samples.

\end{itemize}
	
\end{abstract}

\newpage
\pagenumbering{arabic}
\setcounter{page}{1}


\section{Introduction}
\label{sec:intro}

One of the most basic and natural ways that a probability distribution ${\cal D}$ can be altered is by \emph{truncating} it, i.e.~conditioning on some subset of possible outcomes. 
Indeed, the study of truncated distributions is one of the oldest topics in probability and statistics: already in the 19th century, Galton \cite{Galton97} attempted to estimate the mean and standard deviation of the running times of horses on the basis of sample data that did not include data for horses that were slower than a particular cutoff value. 
Since the running times were assumed to be normally distributed, this was an early attempt to infer the parameters of an unknown normal distribution given samples from a truncated version of the distribution. 
Subsequent early work by other statistical pioneers applied the method of moments \cite{Pearson02,Lee14} and maximum likehood techniques \cite{Fisher31} to the same problem of estimating the parameters of an unknown univariate normal distribution from truncated samples.
The study of truncation continues to be an active area in contemporary statistics (see  \cite{Schneider86,BC14,Cohen16} for recent books on this topic).

Quite recently, a number of research works in theoretical computer science have tackled various algorithmic problems that deal with \emph{high-dimensional} truncated data.  
Much of this work attempts to \emph{learn} a parametric description of an unknown distribution that has been subject to truncation. 
For example, in \cite{DGTZcolt19} Daskalakis, Gouleakis, Tzamos and Zampetakis gave an efficient algorithm for high-dimensional truncated linear regression, and in \cite{DGTZ} Daskalakis, Gouleakis, Tzamos and Zampetakis gave an efficient algorithm for estimating the mean and covariance of an unknown multivariate normal distribution from truncated samples given access to a membership oracle for the truncation set.
In \cite{FKTcolt20} Fotakis, Kalavasis and Tzamos gave a similar result for the setting in which the unknown background distribution is a product distribution over $\zo^d$ instead of a multivariate normal distribution, and in \cite{Kontonis2019} Kontonis, Tzamos and Zampetakis extended the results of \cite{DGTZ} to the case of an unknown truncation set satisfying certain restrictions.
In summary, the recent work described above has focused on learning (parameter estimation) of truncated normal distributions and product distributions over $\zo^d.$

\medskip

\noindent {\bf This Work:  \emph{Detecting} Truncation.}  
In the current paper, rather than the learning problem we study what is arguably the most basic problem that can be considered in the context of truncated data --- namely,  detecting whether or not truncation has taken place at all.  
A moment's thought shows that some assumptions are required in order for this problem to make sense: for example, if the truncation set is allowed to be arbitrarily tiny (so that only an arbitrarily small fraction of the distribution is discarded by truncation), then it can be arbitrarily difficult to detect whether truncation has taken place.  
It is also easy to see that truncation cannot be detected if the unknown truncation set is allowed to be arbitrarily complex.  
Thus, it is natural to consider a problem formulation in which there is a fixed  class of possibilities for the unknown truncation set; this is the setting we consider.

We note that the truncation detection problem we consider has a high-level resemblance to the standard \emph{hypothesis testing} paradigm in statistics, in which the goal is to distinguish between a ``null hypothesis'' and an ``alternate hypothesis.''  In our setting the null hypothesis corresponds to no truncation of the known distribution having taken place, and the alternate hypothesis is that the known distribution has been truncated by some unknown truncation set belonging to the fixed class of possibilities.  However, there does not appear to be work in the statistics literature which deals with the particular kinds of truncation problems considered in this work, let alone computationally efficient algorithms for those problems.

\medskip

\noindent {\bf Prior Work: Convex Truncation of Normal Distributions.}
Recent work \cite{DNS23-convex} considered the truncation detection problem in a setting where the background distribution ${\cal D}$ is the standard multidimensional normal distribution $N(0,1)^n$ and the truncation set is assumed to be an unknown \emph{convex set} in $\R^n$.
This specific problem formulation enabled the use of a variety of sophisticated tools and results from high-dimensional convex geometry, such as Gaussian isoperimetry and the Brascamp-Lieb inequality \cite{BrascampLieb:76} and extensions thereof due to Vempala \cite{Vempala2010}. Using these tools, \cite{DNS23-convex} gave several different algorithmic results and lower bounds. 
Chief among these were (i) a polynomial-time algorithm that uses $O(n/\eps^2)$ samples and distinguishes the non-truncated standard normal distribution $N(0,1)^n$ from $N(0,1)^n$ conditioned on a convex set of Gaussian volume at most $1-\eps$; and (ii) a $\tilde{\Omega}(\sqrt{n})$-sample lower bound for detecting truncation by a convex set of constant volume.

The results of \cite{DNS23-convex} provide a ``proof of concept'' that in sufficiently well-structured settings it can sometimes be possible to detect truncation in a computationally and statistically efficient way.  
This serves as an invitation for a more general study of truncation detection; in particular, it is natural to ask whether strong structural or geometric assumptions like those made in  \cite{DNS23-convex} (normal distribution, a convex truncation set) are required in order to achieve nontrivial algorithmic results. 
Can efficient algorithms detect truncation for broader classes of ``background'' distributions beyond the standard normal distribution, or for  other natural families of truncations besides convex sets?  This question is the motivation for the current work.





\medskip

\noindent
{\bf This Work:  ``Low-Degree'' Truncation of Hypercontractive Product Distributions.}
In this paper we consider 
\begin{itemize}

\item A broader range of possibilities for the background distribution ${\cal D}$ over $\R^n$, encompassing many distributions which may be either continuous or discrete; and

\item A family of non-convex truncation sets corresponding to \emph{low-degree polynomial threshold functions}.

\end{itemize}

Recall that a Boolean-valued function $f: \R^n \to \zo$ is a \emph{degree-$d$ polynomial threshold function} (PTF) if there is a real multivariate polynomial $p(x)$ with $\deg(p) \leq d$ such that $f(x)=1$ if and only if $p(x) \geq 0.$
Low-degree polynomial threshold functions are a well-studied class of Boolean-valued functions which arise naturally in diverse fields such as computational complexity, computational learning theory, and unconditional derandomization, see e.g.~\cite{BruckSmolensky:92,GL:94,HKM2014,diakonikolas2014average,Kane:Gotsman14,DOSW:11,CDS20soda,DKPZ21,BEHYY22,DKNfocs10,Kane11focs,Kane11ccc,Kane12,KaneMeka13,MZ13,Kane14-subpoly,DDS14,DS14,Kane15,KKL17,KL18,KR18,ST18random,BV19,OST20} among many other references.

Our main results, described in the next subsection, are efficient algorithms and matching information-theoretic lower bounds for detecting truncation by low-degree polynomial threshold functions for a wide range of background distributions and parameter settings. 

\subsection{Our Results}

To set the stage for our algorithmic results, we begin with the following observation: 

\begin{observation} \label{obs:generic-upper-bound}
For any fixed (``known'') background distribution ${\cal P}$ over $\R^n$, $O_d(n^{d}/\eps^2)$ samples from an unknown distribution ${\cal D}$ are sufficient to distinguish (with high probability) between the two cases that (i) ${\cal D}$ is the original distribution ${\cal P}$, versus (ii) ${\cal D}$ is ${\cal P}|_{f}$, i.e. ${\cal P}$ conditioned on  $f^{-1}(1)$, where $f$ is an unknown degree-$d$ PTF satisfying $\Pr_{\bx \sim {\cal P}}[f(\bx)=1] \leq 1-\eps$. 
\end{observation}

This is an easy consequence of a standard uniform convergence argument using the well-known fact that the Vapnik-Chervonenkis dimension of the class of all degree-$d$ polynomial threshold functions over $\R^n$ is $O(n^d)$. For the sake of completeness, we give a proof in \Cref{appendix:baseline-distinguisher}.

While the above observation works for any fixed background distribution ${\cal P}$, several drawbacks are immediately apparent. One is that a sample complexity of $O(n^d)$ is quite high, in fact high enough to information-theoretically learn an unknown degree-$d$ PTF; are this many samples actually required for the much more modest goal of merely \emph{detecting} whether truncation has taken place?  A second and potentially more significant issue is that the above VC-based algorithm is computationally highly inefficient, involving a brute-force enumeration over ``all'' degree-$d$ PTFs; for a sense of how costly this may be, recall that even in the simple discrete setting of the uniform distribution over the Boolean cube $\{-1,1\}^n$ there are $2^{\Omega(n^{d+1})}$ distinct degree-$d$ PTFs over  $\{-1,1\}^n$ for constant $d$ \cite[Theorem~2.34]{Saks:93}.  So it is natural to ask whether there exist more efficient (either in terms of running time or sample complexity) algorithms for interesting cases of the truncation detection problem, and to ask about lower bounds for this problem.

On the lower bounds side, it is natural to first consider arguably the simplest case, in which ${\cal P}$ is the uniform distribution ${\cal U}$ over the Boolean hypercube $\bn$.  In this setting we have the following observation:

\begin{observation} \label{obs:simple-lower-bound}
If the truncating PTF $f$ is permitted to have as few as $n^{d/2}$ satisfying assignments, then any algorithm that correctly decides whether its samples come from ${\cal D}={\cal U}$ versus from ${\cal D} = {\cal U}|_f$ must use $\Omega(n^{d/4})$ samples.
\end{observation}

This lower bound can be established using only basic linear algebra and simple probabilistic arguments; it is inspired by the ``voting polynomials'' lower bound of Aspnes et al.~\cite{ABF+:94} against  $\mathsf{MAJ}$-of-$\mathsf{AC}^0$ circuits. We give the argument in \Cref{appendix:weak-lower-bound}.

Taken together, there is a quartic gap between the (computationally inefficient) upper bound given by \Cref{obs:generic-upper-bound} and the information-theoretic lower bound of \Cref{obs:simple-lower-bound} for PTFs with extremely few satisfying assignments.  Our main result is a proof that the true complexity of the truncation distinguishing problem lies exactly in the middle of these two extremes.  We 

\begin{itemize}

\item [(i)] Give a \emph{computationally efficient} distinguishing algorithm which has sample complexity $O(n^{d/2})$ for a wide range of product distributions and values of $\vol(f)$, and 

\item [(ii)] Show that even for the uniform background distribution ${\cal U}$ over $\bn$, distinguishing whether or not ${\cal U}$ has been truncated by a degree-$d$ PTF of volume $\approx 1/2$ requires $\Omega(n^{d/2})$ samples.

\end{itemize}

We now describe our results in more detail.

\medskip

\noindent {\bf An Efficient Algorithm.}  We give a truncation distinguishing algorithm which succeeds if ${\cal P}$ is any multivariate i.i.d.~product distribution ${\cal P} = \mu^{\otimes n}$ over $\R^n$ satisfying a natural \emph{hypercontractivity} property and if $\vol(f)$ is ``not too small.''  We defer the precise technical definition of the (fairly standard) hypercontractivity property that we require to \Cref{subsec:hypercontractive}, and here merely remark that a wide range of i.i.d.~product distributions satisfy the required condition, including the cases where $\mu$ is

\begin{itemize}

\item any fixed distribution over $\R$ that is supported on a finite (independent of $n$) number of points;

\item any normal distribution $N(c,\sigma^2)$ where $c,\sigma$ are independent of $n$;

\item any uniform distribution over a continuous interval $[a,b]$;

\item any distribution which is supported on an interval $[a,b]$ for which there are two constants $0<c<C$ such that everywhere on $[a,b]$ the pdf is between $c/(b-a)$ and $C/(b-a)$.

\end{itemize}

\noindent
An informal statement of our main positive result is below:

\begin{theorem} [Efficiently detecting PTF truncation, informal theorem statement] 
\label{thm:main-positive}
Let $0 < \eps < 1$. Fix any constant $d$ and any hypercontractive i.i.d.~product distribution $\mu^{\otimes n}$ over $\R^n$.
Let $f: \R^n \to \{0,1\}$ be an unknown degree-$d$ PTF such that 
\[1-\epsilon \geq \Prx_{\bx \sim \mu^{\otimes n}}[f(\bx)=1] \geq 2^{-n^{1/\Theta(d)}}.\]
There is an efficient algorithm that uses $\Theta(n^{d/2}/\eps^2)$ samples from ${\cal D}$ and successfully (w.h.p.) distinguishes between the following two  cases: 
\begin{enumerate}
	\item[(i)] ${\cal D}$ is $\mu^{\otimes n}$, i.e. the ``un-truncated'' distribution;  versus
	\item[(ii)] ${\cal D}$ is $\mu^{\otimes n}|_{f}$, i.e.~$\mu^{\otimes n}$ truncated by $f$.
\end{enumerate}
\end{theorem}

Note that $\eps$ is a lower bound on the probability mass of the distribution $\mu^{\otimes n}$ which has been ``truncated;'' as remarked earlier, without a lower bound on $\eps$, it can be arbitrarily difficult to distinguish the truncated distribution. 
Thus, as long as the background distribution is a ``nice'' i.i.d.~product distribution and the truncating PTF's volume is ``not too tiny'', in polynomial time we can achieve a square-root improvement in sample complexity over the naive brute-force computationally inefficient algorithm. 

\medskip

\noindent {\bf A Matching Lower Bound.}
It is natural to wonder whether \Cref{thm:main-positive} is optimal:  can we establish lower bounds on the sample complexity of determining whether a ``nice'' distribution has been truncated by a PTF? And can we do this when the truncating PTF (unlike in \Cref{obs:simple-lower-bound}) has volume which is not extremely small?  

Our main lower bound achieves these goals; it shows that even for the uniform distribution ${\cal U}$ over $\{-1,1\}^n$ and for PTFs of volume $\approx 1/2$, the sample complexity achieved by our algorithm in \Cref{thm:main-positive} is best possible up to constant factors.

\begin{theorem} [Lower bound for detecting PTF truncation, informal theorem statement] 
\label{thm:main-negative}
Fix any constant $d$.
Let $f: \{-1,1\}^n \to \R$ be an unknown degree-$d$ PTF such that $\Pr_{\bx \sim {\cal U}}[f(\bx)=1] \in [0.49,0.51].$
Any algorithm that uses samples from ${\cal D}$ and successfully (w.h.p.) distinguishes between the cases that (i)  ${\cal D}$ is ${\cal U}$,  versus (ii) ${\cal D}$ is ${\cal U}|_{f}$, must use $\Omega(n^{d/2})$ samples.
\end{theorem}

\subsection{Techniques}

We now give a technical overview of both the upper bound (\Cref{thm:main-positive}) and the lower bound (\Cref{thm:main-negative}), starting with the former. 

\subsubsection{Overview of \Cref{thm:main-positive}}

For simplicity, we start by considering the case when the background distribution $\calP = \mun$ is the uniform measure on the Boolean hypercube.  

\medskip

\noindent \textbf{The Boolean Hypercube $\bn$.} Let us denote the uniform measure over $\bn$ by $\calU_n$. 
Recall that our goal is to design an algorithm with the following performance guarantee: Given i.i.d. sample access to an unknown distribution $\calD$, the algorithm w.h.p. 
\begin{itemize}
	\item[(i)] Outputs ``un-truncated'' when $\calD = \calU_n$; and 
	\item[(ii)] Outputs ``truncated'' when $\calD = \calU_n\mid_{f^{-1}(1)}$ for a degree-$d$ PTF $f:\bn\to\zo$, where $1-\eps \geq \Pr[f(\bx) = 1]$ for $\bx\sim\calU_n$.
\end{itemize}
To avoid proliferation of parameters, we set $\eps = 0.1$ for the rest of the discussion. We thus have 
\[\Prx_{\bx\sim\calU_n}\sbra{f(\bx) = 1} \leq 0.9.\]
For any point $x\in\bn$, let $\wt{x} \in \bits^{{n\choose 1} + \ldots + {n\choose d}}$ be the vector given by 
\[\wt{x} := \pbra{\wt{x}_\alpha}_{\substack{\alpha\sse[n] \\ 0 < |\alpha| \leq d}} \qquad \text{where~}\wt{x}_\alpha := \prod_{i\in\alpha} x_i.\]
In other words, every coordinate of $\wt{x}$ corresponds to a non-constant monomial in $x$ of (multilinear) degree at most $d$. Note that the map $x\mapsto \wt{x}$ can be viewed as a \emph{feature map} corresponding to the ``polynomial kernel'' in learning theory. 

The main idea underlying our algorithm, which is given in \Cref{thm:main-positive}, is the following:
\begin{enumerate}
	\item When $\calD = \calU_n$, then it is easy to see that $\Ex_{}\sbra{\wt{\bx}} = \overline{0}$ (the all-$0$ vector). This is immediate from the fact that the expectation of any non-constant monomial under $\calU_n$ is $0$. 
	\item On the other hand, suppose $\calD = \calU_n\mid_{f^{-1}(1)}$ for a degree-$d$ PTF $f$ as above. In this case, it can be shown that 
	\[\vabs{\Ex_{\calU_n\mid_{f^{-1}(1)}}\big[\wt{\bx}\big]}_2 \geq 2^{-\Theta(d)} =: c_d.\]
	This is done by relating the quantity in the LHS above to the Fourier spectrum of degree-$d$ PTFs, which has been extensively studied in concrete complexity theory~(see for example \cite{GL:94,diakonikolas2014average, HKM2014, kane2013correct}). In particular, we obtain this lower bound on $\vabs{\Ex\left[\wt{\bx}\right]}_2$ from an anti-concentration property of low-degree polynomials over the Boolean hypercube. This in turn is a consequence of hypercontractivity of the uniform measure over $\bn$, a fundamental tool in discrete Fourier analysis~(see Section 9.5 of \cite{ODonnell2014}).
\end{enumerate}

Items 1 and 2 above together imply that estimating $\vabs{\Ex\sbra{\wt{\bx}}}_2^2$ up to an additive error of $\pm c_d^2/2$ suffices to distinguish between $\calD = \calU_n$ and $\calD = \calU_n\mid_{f^{-1}(1)}$. Next, note that 
\[\vabs{\Ex_{\calD}\big[\wt{\bx}\big]}_2^2 = \Ex_{\bx,\by\sim\calD}\sbra{\abra{\wt{\bx},\wt{\by}}}.\]
Using the idea of ``U-statistics''~\cite{hoeffding1994class}, this suggests a natural unbiased estimator, namely drawing $2T$ points $\tx{1}, \ldots, \tx{T}$ and $\ty{1},\ldots, \ty{T}$ for some $T$ which we will fix later, and then setting 
\[\bM := \frac{1}{T^2}\abra{\sum_{i=1}^T \tx{i}, \sum_{j=1}^T \ty{j}}.\]
In particular, we have $\Ex[\bM] = \vabs{\Ex\sbra{\wt{\bx}}}_2^2$. 

In order to be able to distinguish between the un-truncated and truncated distributions by estimating $\bM$ (and then appealing to Chebyshev's inequality), it therefore suffices to upper bound the variance of $\bM$ in both the un-truncated and the truncated settings. When $\calD = \calU_n$, then 
$\Var[\bM]$ is straightforward to calculate and it turns out that
\[\Varx_{\calU_n}[\bM] = \frac{m}{T^2} \qquad\text{where~} m:= \#\{\alpha\subseteq[n] : 0 < |\alpha| \leq d\} = O_d(n^d).\]
However, in the truncated setting, $\Var[\bM]$ is significantly trickier to analyze; at a high-level, our analysis expresses $\Var[\bM]$ in terms of the ``weights'' of various levels of the Fourier spectrum of $f$.\footnote{See \Cref{subsec:fourier-101} for a formal definition.} The key technical ingredient we use to control the variance is the so-called ``level-$k$ inequality'' for Boolean functions, which states that for any Boolean function $f:\bn\to\zo$, writing $\bW^{=k}[f]$ for the ``Fourier weight at level-$k$'', we have 
\[\bW^{=k}[f] \leq O_k\pbra{\bW^{=0}[f]\cdot \log^k\pbra{\frac{1}{\bW^{=0}[f]}}}.\]
Recall that $\bW^{=0}[f] = \Ex_{\calU_n}[f]^2$, and so the level-$k$ inequality bounds higher-level Fourier weight in terms of the mean of the function. We remark that the level-$k$ inequality is also a consequence of hypercontractivity over the Boolean hypercube (as before, see Section 9.5 of \cite{ODonnell2014}). With this in hand, we can show that 
\[\Varx_{\calU_n\mid_{f^{-1}(1)}}[\bM] = \frac{O_d(n^d)}{T^2}\]
as long as $\Ex_{\calU_n}[f] \geq 2^{-n^{1/2(2d+1)}}$. 

Finally, taking $T = \Theta_d(n^{d/2})$ implies that the standard deviation of each of our estimators is comparable to the difference in means (which was $c_d^2$), allowing us to distinguish between the un-truncated and truncated settings.

\medskip
\noindent\textbf{Hypercontractive Product Distributions $\mun$.} We use the same high-level approach (as well as the same estimator $\bM$) in order to distinguish low-degree truncation of a hypercontractive product measure $\mun$, but the analysis becomes  more technically involved. To explain the principal challenge, note that over the Boolean hypercube $\bn$, the Fourier basis functions $(\chi_\alpha)_{\alpha\sse[n]},$ 
\[\chi_\alpha(x) := \prod_{i\in \alpha}x_i,\]
form a multiplicative group. This group structure is useful because it means that the product of two basis functions is  another basis function: For $\alpha, \beta \sse [n]$, we have the product formula
\[\chi_\alpha \cdot \chi_\beta = \chi_{\alpha \,\triangle\, \beta}\]
where $\alpha\,\triangle\,\beta$ denotes the symmetric difference of $\alpha$ and $\beta$. 

Over an arbitrary hypercontractive measure $\mun$, this may no longer be the case; as a concrete example, this fails for the Gaussian measure and the Hermite basis (cf. Chapter 11 of~\cite{ODonnell2014}). Over a general product space $\mun$ the Fourier basis functions are now indexed by \emph{multi-subsets} of $[n]$ (as opposed to subsets of $[n]$ over $\bn$)---see the discussion following~\Cref{def:Fourier-basis}. More importantly, there is no simple formula for the product of two Fourier basis functions, and this makes the analysis technically more involved. We remark that this problem, which is known as the \emph{linearization problem},  has been well studied for various classes of orthogonal polynomials (see Section 6.8 of~\cite{andrews1999special}).  \Cref{lem:linearization,lem:ub-linearization-coeffs} establish a weak version of a ``product formula'' between two Fourier basis functions for $\mun$ that suffices for our purposes and lets us carry out an analysis similar to the above sketch for the Boolean hypercube $\bn$.
 

\subsubsection{Overview of \Cref{thm:main-negative}}
 
We turn to an overview of our lower bound, \Cref{thm:main-negative}. As in the previous section, we write $\calU_n$ to denote the uniform distribution over the $n$-dimensional Boolean hypercube $\bn$ and $(\chi_S)_S$ for the Fourier basis over $\bn$. 
To prove \Cref{thm:main-negative}, it suffices to construct a distribution ${\cal F}_d$ over degree-$d$ PTFs over $\bn$ with the following properties: 
 \begin{enumerate}
 \item The distribution ${\cal F}_d$ is supported on thresholds of homogenous degree-$d$ polynomials over $\bits^n$. Note that such polynomials are necessarily multilinear; in particular, each PTF $\boldf \sim {\cal F}_d$ can be expressed as 
 \[\boldf(x) := \mathbf{1}\cbra{\sum_{S:|S|=d} \widehat{\bp}(S) \chi_S(x) \geq 0}.\]
The coefficients $\wh{\bp}(S)$ will be i.i.d.~random variables drawn from the standard Gaussian distribution $N(0, 1)$.
 \item Let $m=\Omega(n^{d/2})$ and consider the distributions
	\begin{itemize}
		\item $\calD_1$, obtained by drawing $m$ independent samples from $\calU_n$; and
		\item $\calD_2$, obtained by first drawing $\boldf\sim\calF_d$, and then drawing $m$ independent samples from $\calU_n\mid_{\boldf^{-1}(1)}$. 
	\end{itemize}
	Then distributions $\calD_1$ and $\calD_2$ are $o(1)$-close to each other in variation distance. 
 \end{enumerate}
 
\noindent Polynomials of the form 
\[\sum_{S} \wh{\bp}(S)\chi_S(x) \qquad\text{for}~\wh{\bp}(S)\sim N(0,1)\] 
are known in the literature as \emph{Gaussian random polynomials}, and have been extensively studied (with an emphasis on the behavior of their roots)~\cite{ibragimov1997roots, hammersley1956zeros, bloom-complex}. We will however be interested in a certain ``pseudorandom-type'' behavior of these polynomials.
 
In particular, we first reduce the problem of proving indistinguishability of ${\cal D}_1$ and ${\cal D}_2$ to proving the following: Suppose ${\bu}_1, \ldots, {\bu}_{m}$ are $m$ randomly chosen points from $\{-1,1\}^n$ (which we fix). Then, with probability $1-o(1)$ over the choice of these $m$ points, 
the distribution of 
\[\pbra{\boldf(\bu_1),\dots,\boldf(\bu_m)}\qquad\text{is $o(1)$-close to that of}\qquad \pbra{\bb_1,\ldots,\bb_m}\]
for $\boldf\sim\calF_d$ and where each $\bb_i$ is an independent unbiased random bit. In other words, we aim to show that if the evaluation points $\bu_1, \ldots, \bu_m$ are randomly chosen (but subsequently known to the algorithm), then $\boldf(\bu_1),\dots,\boldf(\bu_m)$ is $o(1)$-indistinguishable from random. 
 
We establish this last statement by proving something even stronger. Namely, we first observe that the $\mathbb{R}^m$-valued random variable 
 $(\bp(\bu_m), \ldots, \bp(\bu_m))$ is an $m$-dimensional normal random variable for any fixed outcome of $\bu_1, \ldots, \bu_m$. Subsequently, we show that this random variable $(\bp(\bu_m), \ldots, \bp(\bu_m))$  is $o(1)$-close to the standard $m$-dimensional normal random variable $N(0,I_m)$ where $I_m$ is the identity matrix in $m$ dimensions. This exploits a recent bound on total variation distance between multivariate normal distributions~\cite{DMR20} in terms of their covariance matrices, and involves bounding the trace of the Gram matrix generated by random points on the hypercube; details are deferred to the main body of the paper. 

\subsection{Related Work}

As mentioned earlier, ``truncated statistics" has been a topic of central interest to statisticians
for more than a century and recently  in theoretical computer science as well. Starting with the work of Daskalakis {et~al.}~\cite{DGTZ}, several works have looked at the problem of learning an unknown high-dimensional distribution in settings where the algorithm only gets samples from a truncated set~\cite{FKTcolt20, Kontonis2019, bhattacharyya2021efficient}. We note here that in the recent past, there have also been several works on truncation in the area of statistics related to supervised learning scenarios~\cite{daskalakis2021efficient, DGTZcolt19, daskalakis2020truncated}, but the models and techniques in those works are somewhat distant from the topic of the current paper. 
Finally, in retrospect, some earlier works on ``learning from positive samples"~\cite{CDS20soda, DDS15, DGL05} also have a similar flavor. In particular, the main result of  \cite{CDS20soda} is a $\poly(n)$ time algorithm which, given 
 access to samples from a Gaussian truncated by an unknown degree-$2$ PTF, approximately recovers the truncation set; and one of the main results of \cite{DDS15} is an analogous $\poly(n)$-time algorithm but for degree-1 PTF (i.e.~LTF) truncations of the uniform distribution over $\{-1,1\}^n$.  Note that while the settings of \cite{CDS20soda,DDS15} are somewhat related to the current paper, the goals and results of those works are quite different; in particular, the focus is on learning (as opposed to testing / determining whether truncation has taken place), and the sample complexities of the algorithms in \cite{CDS20soda,DDS15}, albeit polynomial in $n$, are polynomials of high degree.

In terms of the specific problem we study, the work most closely related to the current paper is that of \cite{DNS23-convex}.  In particular, as noted earlier, in \cite{DNS23-convex}, the algorithm gets access to samples from either (i) $N(0,1)^n$ or (ii) $N(0,1)^n$ conditoned on a convex set. Besides the obvious difference in the truncation sets which are considered---convex sets in \cite{DNS23-convex} vis-a-vis PTFs in the current paper---the choice of the background distribution in \cite{DNS23-convex} is far more restrictive. Namely, \cite{DNS23-convex} requires the background distribution to be the normal distribution $N(0,1)^n$, whereas the results in current paper hold for the broad family of hypercontractive product distributions (which includes many other distributions as well as the normal distribution). The difference in the problem settings is also reflected in the techniques employed in these two  papers. In particular, the algorithm and analysis of \cite{DNS23-convex} heavily rely on tools from convex geometry including Gaussian isoperimetry~\cite{Borell:85}, the Brascamp-Lieb inequality~\cite{BrascampLieb:76b, BrascampLieb:76}, and recent structural results for convex bodies over Gaussian space~\cite{De2021,DNS22}. In our setting, truncation sets defined by PTFs even of degree two need not be convex, so we must take a very different approach. The algorithm in the current paper uses techniques originating from the study of PTFs in concrete complexity theory, in particular on the hypercontractivity of low-degree polynomials, anti-concentration, and the ``level-$k$'' inequalities~\cite{ODonnell2014}.  So to summarize the current work vis-a-vis \cite{DNS23-convex}, the current work studies a different class of truncations under a significantly less restrictive assumption on the background distribution, and our main algorithm, as well as its analysis, are completely different from those of \cite{DNS23-convex}.


Our lower bound argument extends and strengthens a $\tilde{\Omega}(n^{1/2})$ lower bound, given in \cite{DNS23-convex}, for distinguishing the standard normal distribution $N(0,1)^n$ from $N(0,1)^n|_{f^{-1}(1)}$ where $f$ is an unknown origin-centered LTF (i.e.~a degree-1 PTF);  both arguments use a variation distance lower bound between a standard multivariate normal distribution and a multivariate normal distribution with a suitable slightly perturbed covariance matrix.  Our lower bound argument in the current paper combines tools from the LTF lower bound mentioned above with ingredients (in particular, the use of a ``shadow sample''; see \Cref{claim:lb-reduction}) from a different lower bound from \cite{DNS23-convex} for symmetric slabs; extends the \cite{DNS23-convex} analysis from degree-1 to degree-$d$ for any constant $d$; and gives a tighter analysis than \cite{DNS23-convex} which does not lose any log factors.

We end this discussion of related work with the following overarching high-level question, which we hope will be investigated in future work: Suppose ${\cal P}$ is a background distribution and ${\cal F}$ is a class of Boolean functions. {\em Under what conditions can we distinguish between ${\cal D} = {\cal P}$ versus ${\cal D} = {\cal P}|_f$ (for some $f \in \mathcal{F}$) with sample complexity asymptotically smaller than the sample complexity of learning $\mathcal{F}$?} We view our results on distinguishing truncation by PTFs as a step towards answering this question.

\section{Preliminaries}
\label{sec:prelims}

We write $\N := \{0,1,\ldots\}$ and $\mathbf{1}\{\cdot\}$ for the $0/1$ indicator function. We will write 
\[{[n]\choose d} := \{S \sse[n] : |S| = d\}.\]

Let $(\R, \mu)$ be a probability space. For $n\in\N$, we write $L^2(\R^n, \mun)$ for the (real) inner-product space of functions $f:\R^n\to\R$ with the inner product 
\[\abra{f,g} := \Ex_{\bx\sim{\mun}}\sbra{f(\bx)\cdot g(\bx)}.\]
Here $\mun$ denotes the product probability distribution on $\R^n$. 
{For $q > 0$}
we write 
\[\|f\|_q := \Ex_{\bx\sim{\mun}}\sbra{|f(\bx)|^q}^{1/q}.\]
In particular, for $f:\R^n\to\{0,1\}$, we write $\vol(f) := \|f\|_1 = \Ex[f(\bx)]$ where $\bx\sim{\mun}$. 

We say that a function $f:\R^n\to\zo$ is a \emph{degree-$d$ polynomial threshold function (PTF)} if there exists a polynomial $p:\R^n\to\R$ of degree at most $d$ such that 
\[f(x) = \Indicator\{p(x) \geq 0\}.\]
The primary class of distributions we will consider throughout is that of truncations of an i.i.d. product distribution $\mun$ by a degree-$d$ PTF {of at least some minimum volume}; more precisely, we will consider the following class of truncations:
\begin{equation} \label{eq:class-def}
	\Pptf\pbra{d, \alpha} := \cbra{\mun\mid_{f^{-1}(1)} \ : f \text{ is a degree-$d$ PTF with } \vol(f) \geq \alpha }
\end{equation}
where $\alpha=\alpha(n)$ may depend on $n$ (in fact we will be able to take $\alpha$ as small as $2^{-\Theta(\sqrt{n})}$).
Throughout the paper we will assume that $d$ (the degree of the PTFs we consider) is a fixed constant independent of the ambient dimension $n$.

\subsection{Harmonic Analysis over Product Spaces}
\label{subsec:fourier-101}

Our notation and terminology in this section closely follow those of O'Donnell~\cite{ODonnell2014}; in particular, we refer the reader to Chapter~8 of~\cite{ODonnell2014} for further background. 

\begin{definition}~\label{def:Fourier-basis}
	A \emph{Fourier basis} for $L^2(\R, \mu)$ is an orthonormal basis $\calB = \{\chi_0,\chi_1, \ldots\}$ with $\chi_0 \equiv 1$. 
\end{definition}

It is well known that if $L^2(\R, \mu)$ is separable,\footnote{Recall that a metric space is \emph{separable} if it contains a countable dense subset.} then it has a Fourier basis (see for e.g. Section I.4 of \cite{conway2019course}). Note that we can obtain a Fourier basis for $L^2(\R^n, \mun)$ by taking all possible $n$-fold products of elements of $\calB$; more formally, for a multi-index $\alpha\in\N^n$, we define 
\[\chi_{\alpha}(x) :=\prod_{i=1}^n \chi_{\alpha_i}(x_i).\]
Then the collection 
$\calB_n := \cbra{ \chi_\alpha : \alpha_i \in \N^n}$
forms a Fourier basis for $L^2(\R^n, \mun)$; this lets us write $f\in L^2(\R^n, \mun)$ as $f = \sum_{\alpha\in\N^n} \wh{f}(\alpha)\chi_{\alpha}$ where  
\[\wh{f}(\alpha) := \abra{f, \chi_\alpha}\] 
is the \emph{Fourier coefficient of $f$ on $\alpha$}.

We can assume without loss of generality that the basis elements of $L^2(\R, \mu)$, namely $\cbra{\chi_0, \chi_1, \ldots}$, are polynomials with $\deg(\chi_i) = i$. This is because a polynomial basis can be obtained for $L^2(\R, \mu)$ by running the Gram-Schmidt process. By extending this basis to $L^2(\R^n,\mun)$ by taking products, it follows that we may assume without loss of generality that for a multi-index $\alpha \in \N^n$, we have $\deg(\chi_\alpha) = |\alpha|$ where 
\[|\alpha| := \sumi \alpha_i.\]
We will also write $\#\alpha := \abs{\supp(\alpha)}$ where $\supp(\alpha):=\cbra{i : \alpha_i \neq 0 }$. 

\begin{remark}
While the Fourier coefficients $\{\wh{f}(\alpha)\}$ depend on the choice of basis $\cbra{\chi_\alpha}$, we will always work with some fixed (albeit arbitrary) polynomial basis, and hence there should be no ambiguity in referring to the coefficients as though they were unique. We assume  that the orthogonal basis $\{\chi_\alpha\}$ is ``known'' to the algorithm; this is certainly a reasonable assumption for natural examples of hypercontractive distributions (e.g. distributions with finite support, the uniform distribution on intervals, the Gaussian distribution, etc.), and is in line with our overall problem formulation of detecting whether a known background distribution has been subjected to truncation.
\end{remark}

As a consequence of orthonormality, we get that for any $f\in L^2(\R^n,\mun)$, we have 
\[\Ex_{\bx\sim\mun}\sbra{f(\bx)} = \wh{f}(0^n) \qquad\text{and}\qquad \|f\|_2^2 = \sum_{\alpha\in\N^n} \wh{f}(\alpha)^2,\]
with the latter called \emph{Parseval's formula}.
We also have \emph{Plancharel's formula}, which says that 
\[\abra{f, g} = \sum_{\alpha\in\N^n} \wh{f}(\alpha)\wh{g}(\alpha).\]
Finally, we write 
\[\bW^{=k}[f] := \sum_{|\alpha| = k}\wh{f}(\alpha)^2 \qquad \text{and} \qquad \bW^{\leq k}[f] := \sum_{|\alpha| \leq k}\wh{f}(\alpha)^2\]
for the \emph{Fourier weight of $f$ at level $k$} and the \emph{Fourier weight of $f$ up to level $k$} respectively.

\subsection{Hypercontractive Distributions}
\label{subsec:hypercontractive}

The primary analytic tools we will require in both our upper and lower bounds are consequences of \emph{hypercontractive} estimates for functions in $L^2(\R^n,\mun)$; we refer the reader to Chapters~9 and 10 of~\cite{ODonnell2014} for further background on hypercontractivity and its applications.

\begin{definition} \label{def:hypercontractive-domain}
	We say that $(\R, \mu)$ is \emph{hypercontractive} if for every $q\geq 2$, there is a fixed constant $C_q(\mu)$ such that for every $n \geq 1$ and every multivariate degree-$d$ polynomial $p : \R^n\to\R$ we have 
	\begin{equation} \label{eq:general-bonami}
		\|p\|_q \leq C_q(\mu)^d\cdot \|p\|_2
	\end{equation}
	where $C_q(\mu)$ is independent of $n$ and satisfies
	\begin{equation} \label{eq:hyp-constant-growth}
		C_q(\mu) \leq K\sqrt{q}
	\end{equation}	
	for an absolute constant $K$.
	When the product distribution $\mun$ is clear from context, we will sometimes simply write $C_q := C_q(\mu)$ instead. 
\end{definition}

It is clear from the monotonicity of norms that $C_q \geq 1$; see \Cref{table:hypercontractive} for examples of hypercontractive distributions with accompanying hypercontractivity constants $C_q(\mu)$. 

\begin{remark} \label{remark:annoying-stuff}
We  note that \Cref{def:hypercontractive-domain} is not the standard definition of a hypercontractive product distribution (cf. Chapters~9 and 10 of \cite{ODonnell2014}), but is in fact an easy consequence of hypercontractivity that is sometimes referred to as the ``Bonami lemma.'' The guarantees of \Cref{eq:general-bonami,eq:hyp-constant-growth} are all we require for our purposes, and so we choose to work with this definition instead.
\end{remark}

\begin{remark} \label{rem:hyp-constant-growth}
While \Cref{eq:hyp-constant-growth} may seem extraneous, we note that the ``level-$k$ inequalities'' (\Cref{prop:level-k-inequalities}) crucially rely on this bound on the hypercontractivity constant $C_q$.
\end{remark}

{\renewcommand{\arraystretch}{1.5}
\begin{table}[]
\begin{center}
\begin{tabular}{@{}lll@{}}
\toprule
Product Distribution $\calD$                                         & $C_q$ &   Reference       \\ \midrule
Standard Gaussian Distribution $N(0,1)^n$                    &      $\sqrt{q-1}$      &      Proposition~5.48 of~\cite{Aubrun2017}     \\ 
Uniform Measure on $\bn$                       &    $\sqrt{q-1}$        &       Theorem~9.21 of~\cite{ODonnell2014}    \\
Finite Product Domains $(\Omega^n, \mun)$ &  $\sqrt{\frac{q}{2\min (\mu)}}$          &      \cite{Wolff2007} and Fact~2.8 of~\cite{aushas09}     \\
\bottomrule
\end{tabular}
\caption{Examples of hypercontractive distributions, along with their accompanying hypercontractivity constants. Here $\min(\mu)$ denotes the minimal non-zero probability of any element in the support of the (finitely supported) distribution $\mu$. } 
\label{table:hypercontractive}
\end{center}
\end{table} 
}

We turn to record several useful consequences of hypercontractivity which will be crucial to the analysis of our estimator in \Cref{sec:ub} and our lower bound in \Cref{sec:lb}. 

The following \emph{anti-concentration} inequality is a straightforward consequence of hypercontractivity. A similar result for arbitrary product distributions with finite support was obtained by Austrin--H\aa stad~\cite{aushas09}, and a similar result for functions over $\bn$ with the uniform measure was obtained by Dinur et al.~\cite{DFKO06}. The proof of the following proposition closely follows that of Proposition~9.7 in~\cite{ODonnell2014}:

\begin{proposition}[Anti-concentration of low-degree polynomials] \label{prop:anti-conc}
  Suppose $(\R^n, \mun)$ is a hypercontractive probability space. Then for any degree-$d$ polynomial $p: \R^n\to \R$ with $\Ex[p]=0$ and $\Varx[p] = 1$, we have
  \begin{align*}
    \Prx_{\bx\sim\mun}\sbra{|p(\bx)| \geq \frac{1}{2}} \geq 0.5625 \cdot c^d
  \end{align*} 
  for a constant $c := c(\mu)$ independent of $n$. 
\end{proposition}

\begin{proof}
	As $\Ex\sbra{p(\bx)^2}=1$, it follows by the definition of hypercontractivity that 
	\[\Ex\sbra{p(\bx)^4} \leq C_4(\mu)^{4d}.\]
	Applying the Paley--Zygmund inequality to the non-negative random variable $p(\bx)^2$, we get that 
	\begin{align*}
		\Prx_{\bx\sim\mun}\sbra{|p(\bx)| \geq \frac{1}{2}} & = \Prx_{\bx\sim\mun}\sbra{p(\bx)^2 \geq \frac{1}{4}}\\
		&\geq \pbra{\frac{3}{4}}^2\cdot\pbra{\frac{\Ex\sbra{p(\bx)^2}^2}{\Ex\sbra{p(\bx)^4}}}\\
		&\geq 0.5625\cdot\pbra{\frac{1}{C_4(\mu)^4}}^d,
	\end{align*}
	completing the proof. 
\end{proof}

The following proposition bounds Fourier weight up to level $k$ (i.e. $\bW^{\leq k}[f]$) in terms of the bias (i.e. the degree-$0$ Fourier coefficient) of the function. We note that an analogous result for functions over $\bn$ with the uniform measure is sometimes known as ``Chang's Lemma'' or ``Talagrand's Lemma'' \cite{changs-lemma,Talagrand:96}; see also Section 9.5 of \cite{ODonnell2014}.

\begin{proposition}[Level-$k$ inequalities] \label{prop:level-k-inequalities}
	Suppose $(\R,\mu)$ is hypercontractive and $f : \R^n \to \zo$ is a Boolean function. Then for all $1 \leq k \leq 2\log\pbra{\frac{1}{\vol(f)}}$ we have
	\[\bW^{\leq k}[f] \leq K^{k}\cdot \vol(f)^2\cdot \pbra{\log\pbra{\frac{1}{\vol(f)}}}^k\]
	where $K$ is a constant independent of $n$.
\end{proposition}

\begin{proof}
	Let $f^{\leq k} := \sum_{|\alpha| \leq k} \wh{f}(\alpha)\chi_\alpha$ and note that $\bW^{\leq k}[f] = \langle f, f^{\leq k} \rangle$. By H\"older's inequality, we have
	\begin{equation} \label{eq:level-k-holder}
		\langle f, f^{\leq k} \rangle \leq \|f\cdot f^{\leq k}\|_1 \leq \|f\|_{s}\cdot\|f^{\leq k}\|_{t}
	\end{equation}
	for all $1\leq s, t \leq \infty$ with $1/s + 1/t = 1$. We will take 
	\[
		t :=  \frac{q}{k} \quad \text{where}~q := 2\log\pbra{\frac{1}{\vol(f)}},
	\]
	noting that $t\geq 1$ as $k\leq q$. We next bound each term on the right hand side of \Cref{eq:level-k-holder} separately. First, note that 
	\begin{equation} \label{eq:holder-rhs-first-term}
		\|f\|_{s} \leq \vol(f)^{1/s} = \vol(f)^{1- 1/t} 
	\end{equation}
	as $f(x) \in \zo$ and $s \geq 1$. 
	Furthermore, as $t\leq q$, we have by hypercontractivity that  
	\begin{equation} \label{eq:holdr-rhs-second-term}
		\|f^{\leq k}\|_{t} \leq \|f^{\leq k}\|_q \leq C_q^k\cdot \|f^{\leq k}\|_2 = C_q^k\cdot\sqrt{\bW^{\leq k}[f]}
	\end{equation}
	where the first inequality is due to monotonicity of norms. 
	Combining \Cref{eq:level-k-holder,eq:holder-rhs-first-term,eq:holdr-rhs-second-term} and rearranging, we get that 
	\begin{align*}
		\bW^{\leq k}[f] &\leq C_q^{2k}\cdot \vol(f)^2 \cdot \pbra{\frac{1}{\vol(f)}}^{2/t}\\
		&\ = C_q^{2k}\cdot \vol(f)^2\cdot\pbra{\frac{1}{\vol(f)}}^{k/\log\pbra{\frac{1}{\vol(f)}}}\\
		&= C_q^{2k}\cdot \vol(f)^2 \cdot 2^k.
\end{align*}
Recall from \Cref{eq:hyp-constant-growth} (see also \Cref{rem:hyp-constant-growth}) that $C_q^2 \leq {K_1}^2\cdot q = 2{K_1}^2\log\pbra{\frac{1}{\vol(f)}}$ for an absolute constant $K$, and so we have 
\[\bW^{\leq k}[f] \leq {K}^k\cdot \vol(f)^2 \cdot \pbra{\log\pbra{\frac{1}{\vol(f)}}}^k \qquad\text{where } {K} = 4{K_1}^2,\]
completing the proof. 
\end{proof}

\begin{remark} \label{rem:level-k-general}
	We note that the \Cref{prop:level-k-inequalities} also holds for bounded functions $f:\R^n\to[-1,1]$ with $\vol(f) := \Ex[|f|]$, although we will not require this. 
\end{remark}



\section{An $O(n^{d/2})$-Sample Algorithm for Degree-$d$ PTFs}
\label{sec:ub}

In this section, we present a $O(n^{d/2})$-sample algorithm for distinguishing a hypercontractive product distribution $\mun$ from $\mun$ truncated by the satisfying assignments of a degree-$d$ PTF. More precisely, we prove the following in \Cref{subsec:proof-ub}: 

\begin{theorem} \label{thm:ptf-ub}
Let $\eps > 0$ and let {$(\R,\mu)$} be hypercontractive. 
There is an algorithm, \PTF~(\Cref{alg:ptf}), with the following performance guarantee: Given access to independent samples from any unknown distribution ${\cal D} \in \{\mun\} \cup \Pptf(d, 2^{-n^{1/2(2d+1)}})$, the algorithm uses $T$ samples where 
\[T := \Theta_d\pbra{\frac{n^{d/2}}{\min\cbra{1, \eps/(1-\eps), c^{\Theta(d)}/(1-\eps)}^2}}\]
with $c := c(\mu)$ as in \Cref{prop:anti-conc}, runs in $O_d(T\cdot n^d)$ time, and 
\begin{enumerate}

\item If ${\cal D}=\mun$, then with probability at least $9/10$ the algorithm outputs ``un-truncated;''

\item If $\dtv\pbra{\mun, \calD} \geq \eps$ (equivalently, ${\cal D}=\mun|_f$ for some degree-$d$ PTF $f$ with $2^{-n^{1/2(2d+1)}} \leq \vol(f) \leq 1-\eps$), then with probability at least $9/10$ the algorithm outputs ``truncated.''

\end{enumerate}
\end{theorem}

\begin{algorithm}[t]
\caption{Distinguisher for degree-$d$ PTFs. Throughout the algorithm the constant $c := c(\mu)$ is as in the proof of \Cref{lem:low-degree-weight}.}
\label{alg:ptf}
\vspace{0.5em}
\textbf{Input:} $\calD\in\{\mun\}\cup\Pptf({d, 2^{-\sqrt{n}}})$, $\eps > 0$\\[0.5em]
\textbf{Output:} ``Un-truncated'' or  ``truncated''

\ 

\PTF$(\calD)$:
\begin{enumerate}
    \item Draw $2T$ independent sample points $\x{1}, \ldots, \x{T}, \y{1},\ldots, \y{T} \sim \calD$, where \[T := \Theta_d\pbra{\frac{n^{d/2}}{\min\cbra{1, \eps/(1-\eps), c^{\Theta(d)}/(1-\eps)}^2}}.\] 
    \item Compute the statistic $\bM$ where 
    \[\bM := \frac{1}{T^2} \abra{\sum_{i=1}^{T} \tx{i}, \sum_{j=1}^{T} \ty{j}} \qquad\text{with}\qquad \tx{i} := \pbra{\chi_{\alpha}(\x{i})}_{1\leq|\alpha|\leq d}\]
    and $\ty{j}$ defined similarly. 
    \item Output ``truncated'' if $\bM \geq \min\cbra{1, \pbra{\frac{\epsilon}{1-\epsilon}}, \frac{c^{\Theta(d)}}{(1-\eps)}}$
    	and ``un-truncated'' otherwise.
\end{enumerate}

\end{algorithm}

Before proceeding to the proof of \Cref{thm:ptf-ub}, we give a brief high-level description of \Cref{alg:ptf}. The algorithm draws $2T$ independent samples $\{\x{i}, \y{i}\}_{i\in T}$ where $T$ is as above, and then performs a \emph{feature expansion}
to obtain the $2T$ vectors $\{\tx{i}, \ty{i}\}_{i\in T}$ where 
\[\tx{i} := \pbra{\chi_\alpha(\x{i})}_{1\leq|\alpha|\leq d}\]
and $\ty{i}$ is defined similarly. The statistic $\bM$ employed by the algorithm to distinguish between the un-truncated and truncated is then given by 
\begin{enumerate}
	\item First computing the average of the kernelized $\tx{i}$ vectors and the kernelized $\ty{i}$ vectors; and then
	\item Taking the inner product between the two averaged kernel vectors.
\end{enumerate}
An easy calculation, given below, relates the statistic $\bM$ to the low-degree (but not degree-0) Fourier weight {of the truncation function (note that if no truncation is applied then the truncation function is identically 1)}. The analysis then proceeds by using anti-concentration of low-degree polynomials to show that the means of the estimators differ by $\Omega_\eps(1)$ between the two settings. We bound the variance of the estimator in both the un-truncated and truncated setting (using the level-$k$ inequalities at a crucial point in the analysis of the truncated setting), and given a separation between the means and a bound on the variances, it is straightforward to distinguish between the two settings using Chebyshev's inequality. 

\begin{remark}
We note that the trick of drawing a ``bipartite'' set of samples, i.e. drawing $2T$ samples $\{\x{i}, \y{i}\}_{i\in T}$, was recently employed in the algorithm of Diakonikolas,~Kane,~and~Pensia~\cite{DKP-SOSA} for the problem of Gaussian mean testing. For our problem we could have alternately used the closely related estimator $\bM'$ given by
\[\bM' := {T \choose 2}^{-1} \sum_{i\neq j} \abra{\tx{i},\tx{j}}\]
to distinguish between the un-truncated and truncated distributions via a similar but slightly more cumbersome analysis. {We note that the main technical tool used in the analysis of  Diakonikolas,~Kane,~and~Pensia~\cite{DKP-SOSA} is the  Carbery-Wright anti-concentration inequality for degree-2 polynomials in Gaussian random variables, whereas our argument uses the above-mentioned kernelization approach and other consequences of hypercontractivity, namely the level-$k$ inequalities, beyond just anti-concentration.}
\end{remark}

\subsection{Useful Preliminaries}
\label{subsec:ub-prelims}

The following lemma will be crucial in obtaining a lower-bound for the expectation of our test statistic $\Ex[\bM]$ in the truncated setting; we note that an analogous statement in the setting of the Boolean hypercube was obtained by Gotsman and Linial~\cite{GL:94}.

\begin{lemma} \label{lem:low-degree-weight}
Suppose $(\R, \mu)$ is hypercontractive. If $f:\R^n\to\zo$ is a degree-$d$ PTF, then 
\[\sum_{1\leq|\alpha|\leq d} \wh{f}(\alpha)^2 \geq \Omega\pbra{\min\cbra{\vol(f), 1-\vol(f), c^{\Theta(d)}}}^2\]
for an absolute constant $c := c(\mu) \in (0, 1]$. 
\end{lemma}

\begin{proof}
We may assume that 
\[f(x) = \Indicator\{p(x) \geq \theta \}\]
where $p:\R^n\to\R$ is a degree-$d$ polynomial with $\E[p(\bx)]=\widehat{p}(0^n)=0$ and $\|p\|_2^2 = \Var[p] = \sum_\alpha \wh{p}(\alpha)^2 = 1$. By Cauchy--Schwarz and Plancherel, we get
\begin{align}
	{\sum_{1\leq |\alpha|\leq d} \wh{f}(\alpha)^2} &= {\pbra{\sum_{1\leq |\alpha|\leq d} \wh{f}(\alpha)^2}\pbra{\sum_{1\leq |\alpha|\leq d} \wh{p}(\alpha)^2}} \nonumber \\
	&\geq \pbra{\sum_{1 \leq |\alpha|\leq d} \wh{f}(\alpha)\cdot\wh{p}(\alpha)}^2\nonumber \\
	&= { \pbra{\Ex\sbra{f(\bx)\cdot p(\bx)}}^2.} \label{eq:potato}
\end{align}
where we made use of the fact that $p$ is a degree-$d$ polynomial with $\widehat{p}(0^n)=0$ and $\Var\sbra{p(\bx)} = 1$. Note that by \Cref{prop:anti-conc}, we have that either 
\begin{equation} \label{eq:anti-conc-casework}
	\Prx_{\bx\sim\mun}\sbra{p(\bx) \geq \frac{1}{2}} \geq \Omega\pbra{c^d} \qquad \text{or} \qquad \Prx_{\bx\sim\mun}\sbra{p(\bx) \leq -\frac{1}{2}} \geq \Omega\pbra{c^d}. 
\end{equation}
Suppose that it is the former. (As we will briefly explain later, the argument is symmetric in the latter case.) We further break the analysis into cases depending on the magnitude of $\theta$. 

\paragraph{Case 1: $\theta \geq \frac{1}{2}$.}  In this case, we have by \Cref{eq:potato} that 
\begin{align*}
	{\sum_{1\leq |\alpha|\leq d} \wh{f}(\alpha)^2} &\geq \pbra{\Ex\sbra{f(\bx)\cdot p(\bx)}}^2 \\
	&= \pbra{\Ex\sbra{ \Indicator_{p(\bx)\geq\theta } \cdot p(\bx)}}^2 \nonumber \\
	&\geq (\vol(f)\cdot\theta)^2\\
	&\geq \Omega\pbra{\vol(f)^2},
\end{align*}
and so the result follows.

\paragraph{Case 2: $0 \leq \theta < \frac{1}{2}$.} In this case, we have by \Cref{prop:anti-conc} that
\[\vol(f) = \Prx_{\bx\sim\mun}\sbra{p(\bx) \geq \theta} \geq \Prx_{\bx\sim\mun}\sbra{p(\bx) \geq \frac{1}{2}} \geq \Omega\pbra{c^d}.\]
Once again by \Cref{eq:potato}, we have 
\[{\sum_{1\leq |\alpha|\leq d} \wh{f}(\alpha)^2} \geq \pbra{\Ex\sbra{f(\bx)\cdot p(\bx)}}^2 
\geq \pbra{{\frac 1 2} \Pr\left[p(\bx) \geq {\frac 1 2}\right]}^2\geq \Omega\pbra{c^{\Theta(d)}},
\]
where the second inequality follows from {$f \cdot p$ being always non-negative and at least ${\frac 1 2}$ with probability $\Pr[p(\bx) \geq {\frac 1 2}].$}

\paragraph{Case 3: $\theta < 0$.}

Consider the degree-$d$ PTF $f^{\dagger} := 1-f$ given by 
\[f^\dagger(x) = \mathbf{1}\cbra{p(x) < \theta}.\]
It is easy to check that $|\wh{f^\dagger}(\alpha)| = |\wh{f}(\alpha)|$ for all $S \neq \emptyset$ and that $\vol(f^\dagger) = 1-\vol(f)$. Repeating the above analysis then gives that 
\[\sum_{1\leq|\alpha|\leq d} \wh{f}(\alpha)^2 = \sum_{1\leq|\alpha|\leq d} \wh{f^\dagger}(\alpha)^2 \geq  \Omega\pbra{\vol(f^\dagger)^2\cdot c^{\Theta(d)}} = \Omega\pbra{\pbra{1-\vol(f)}^2\cdot c^{\Theta(d)}}. \]

\

Putting Cases 1 through 3 together, we get that 
\begin{equation}
	\sum_{1\leq|\alpha|\leq d} \wh{f}(\alpha)^2 \geq \Omega\pbra{\min\cbra{\vol(f), 1-\vol(f), c^{\Theta(d)}}}^2,
\end{equation}
completing the proof. Recall, however, that we assumed that 
\[\Prx_{\bx\sim\mun}\sbra{p(\bx) \geq \frac{1}{2}} \geq \Omega\pbra{c^d}\] 
in \Cref{eq:anti-conc-casework}. Suppose that we instead have 
\[\Prx_{\bx\sim\mun}\sbra{p(\bx) \leq -\frac{1}{2}} \geq \Omega\pbra{c^d}.\]
Then note that the same trick used in Case 3 by considering $f^\dagger$ instead of $f$ and repeating the three cases completes the proof. 
\end{proof}

We will also require the following two lemmas which are closely linked to the \emph{linearization problem for orthogonal polynomials} (see Section 6.8 of~\cite{andrews1999special}). The first lemma bounds the magnitude of the Fourier coefficients of the \emph{product} of basis functions; while the estimate below relies on hypercontractivity, we note that exact expressions for the Fourier coefficients are known for various classes of orthogonal polynomials including the Chebyshev, Hermite, and Laguerre polynomials.

\begin{lemma} \label{lem:ub-linearization-coeffs}
	Suppose $(\R, \mu)$ is hypercontractive. Let $\alpha, \beta, \gamma \in \N^n$. Then 
	\[\abra{\chi_\alpha, \chi_\beta\cdot\chi_\gamma} \leq C_4({\mu}\ignore{\mun})^{|\beta| + |\gamma|}.\]
\end{lemma}

\begin{proof}
	Using Cauchy--Schwarz, we have that 
	\begin{align*}
		\abra{\chi_\alpha, \chi_\beta\cdot\chi_\gamma} &= \Ex_{\bx\sim\mun}\sbra{\chi_\alpha(\bx)\cdot\pbra{\chi_\beta(\bx)\cdot\chi_\gamma(\bx)}}\\
		&\leq \sqrt{\Ex_{\bx\sim\mun}\sbra{\chi_{\alpha}(\bx)^2}\Ex_{\bx\sim\mun}\sbra{\chi_{\beta}(\bx)^2\cdot \chi_{\gamma}(\bx)^2}}\\
		& = \sqrt{\Ex_{\bx\sim\mun}\sbra{\chi_{\beta}(\bx)^2\cdot \chi_{\gamma}(\bx)^2}}
		\intertext{due to the orthonormality of $\chi_{\alpha}$. Using Cauchy--Schwarz once again, we have that}
		\abra{\chi_\alpha, \chi_\beta\cdot\chi_\gamma}
		&\leq \pbra{\Ex_{\bx\sim\mun}\sbra{\chi_{\beta}(\bx)^4}\cdot\Ex_{\bx\sim\mun}\sbra{\chi_{\gamma}(\bx)^4}}^{1/4}\\
		&= \|\chi_\beta\|_4\cdot\|\chi_\gamma\|_4\\
		&\leq C_4(\ignore{\mun}\mu)^{|\beta| + |\gamma|}
	\end{align*}
	where the final inequality uses hypercontractivity and 
\ignore{orthonormality of $\chi_\beta$ and $\chi_\gamma$, completing the proof.}the fact that $\chi_\beta$ ($\chi_\gamma$ respectively) is a polynomial of two-norm 1 and degree at most $|\beta|$ (at most $|\gamma|$ respectively).
\end{proof}

Finally, we also require the following combinatorial lemma: 

\begin{lemma} \label{lem:linearization}
	Given $d\in\N$ and a fixed multi-index $\alpha\in\N^n$ with $|\alpha| := k \leq 2d$, we have 
	\[\#\cbra{(\beta, \gamma) \in \N^n\times\N^n : |\beta|,|\gamma| \leq d \text{ and } \abra{\chi_{\alpha}, \chi_{\beta}\chi_{\gamma}} \neq 0 } \leq O_d\pbra{n^{d-\ceil{k/2}}}.\]
\end{lemma}

\begin{proof}
We will upper bound the number of pairs $(\beta, \gamma)$ such that $|\beta|,|\gamma| \leq d$ and $\abra{\chi_{\alpha}, \chi_{\beta}\chi_{\gamma}} \neq 0$; so fix such a pair $(\beta,\gamma)$.
  
  We first note that for any $i\notin\supp(\alpha)$, we must have
  $\beta_i = \gamma_i$; for otherwise $\abra{\chi_{\alpha}, \chi_{\beta}\chi_{\gamma}} = 0$ due to orthonormality
  of $\chi_{\beta_i}$ and $\chi_{\gamma_i}$.
  Also, for each $i\in \supp(\alpha)$, we must have $\beta_i + \gamma_i \geq \alpha_i$,
  since
  \ignore{otherwise as $\chi_{\beta_i}\chi_{\gamma_i} \in \mathrm{span}\cbra{\chi_{j} : j \leq \beta_i + \gamma_i}$. In more detail,} if $\beta_i + \gamma_i < \alpha_i$, then 
  \[
  \ignore{\abra{\beta_i\cdot\gamma_i, \alpha_i}}
\abra{\chi_{\beta_i} \cdot \chi_{\gamma_i}, \chi_{\alpha_i}}
   = 0\]
by orthonormality and the fact that $\chi_{\beta_i} \cdot \chi_{\gamma_i}$ is a linear combination of basis functions $\{\chi_j\}_{j<\alpha_i}.$
\ignore{  as the Fourier transform of $\chi_{\beta_i}\chi_{\gamma_i}$ will only have non-zero Fourier coefficients on polynomials with degree strictly less than $\alpha_i$.
}

  Summing over all $i\in\supp(\alpha)$, we get that 
  \[\sum_{i\in\supp(\alpha)} \beta_i + \gamma_i \geq \sum_{i\in\supp(\alpha)} \alpha_i = |\alpha|,\]
  and so it follows that either 
  \[\sum_{i\in \supp(\alpha)} \beta_i \geq \ceil{k/2} \qquad\text{or}\qquad \sum_{i\in \supp(\alpha)} \gamma_i \geq \ceil{k/2};\]
  without loss of generality we suppose it is the former. 
It follows that the total number of ways of choosing such a pair $(\beta,\gamma)$ is bounded by
\begin{align*}
O_d(1) \cdot O_d(1) \cdot \sum_{j=0}^{d-\lceil k/2 \rceil} 
  (n-|\supp(\alpha)|)^j \leq
  O_d(1) \cdot \sum_{j=0}^{d-\lceil k/2 \rceil} 
  n^j 
  \leq O_d(n^{d-\ceil{k/2}})
\end{align*}
where
\begin{itemize}
  	\item The first two $O_d(1)$ factors bound the number of possible outcomes of $(\beta_i)_{i \in \supp(\alpha)}$ (recall that each $\beta_i \in [0,d]$ and $|\supp(\alpha)| \leq 2d$) and $(\gamma)_{i \in \supp(\alpha)}$ respectively); and
 
 	\item The third term on the LHS upper bounds the number of choices for $\pbra{\beta_i}_{i\notin\supp(\alpha)}$. Recall that for  $i\notin\supp(\alpha)$ we must have
  $\beta_i = \gamma_i$, and so 
  \[\pbra{\beta_i}_{i\notin\supp(\alpha)} = \pbra{\gamma_i}_{i\notin\supp(\alpha)}\]
  so we need only bound the number of choices of $\pbra{\beta_i}_{i\notin\supp(\alpha)}$;
moreover, as
  \[
  	\sum_{i} \beta_i \leq d \quad \text{and } \sum_{i\in\supp(\alpha)} \beta_i \geq \ceil{k/2}, \text{ we must have}  \sum_{i\notin\supp(\alpha)} \beta_i \leq d - \ceil{k/2}, 
  \]
  \end{itemize}
completing the proof. 
\end{proof}

\subsection{Proof of \Cref{thm:ptf-ub}}
\label{subsec:proof-ub}


\begin{proofof}{\Cref{thm:ptf-ub}}
Throughout, we will write 
\[m := \#\cbra{\alpha\in\N^n : 1 \leq |\alpha| \leq d}.\]
Recalling that there are at most ${{n+k-1}\choose k}$ multi-indices $\alpha\in\N^n$ with $|\alpha| = k$, we have that 
\[m\leq O_d(n^d).\]

We compute the mean and variance of the estimator $\bM$ in each case (i.e. when $\calD = \mun$ and when $\calD \in \Pptf(d, 2^{-n^{1/2(2d+1)}})$) separately.  

\paragraph{Case 1: $\calD = \mun$.} For brevity (and to distinguish the calculations in this case from the following one), we denote the un-truncated distribution by 
\[\calD_u := \mun.\]
When $\x{1},\ldots,\x{T},\y{1},\ldots,\y{T}\sim\calD_u$, we have by bi-linearity of the inner product that
\begin{align}
	\Ex_{\calD_u}\sbra{\bM} &= \frac{1}{T^2} \sum_{i,j=1}^{T} \Ex_{\calD_u}\sbra{\abra{\tx{i},\ty{j}}} \nonumber\\
	&=\sum_{1\leq|\alpha|\leq d} \Ex_{\calD_u}\sbra{\chi_\alpha(\bx)\cdot\chi_\alpha(\by)} \nonumber \\
	&= \sum_{1\leq|\alpha|\leq d} \Ex_{\calD_u}\sbra{\chi_\alpha(\bx)}\cdot\Ex_{\calD_u}\sbra{\chi_\alpha(\by)} \nonumber \\
	&= 0 \label{eq:untrunc-means}
\end{align}
as $\bx,\by\sim\calD_u$ are independent samples and also because of the orthonormality of $\{\chi_{\alpha}\}$. 

\ignore{
Turning to the variance, we have 
\begin{align}
	\Varx_{\calD_u}[\bM] = \frac{1}{T^4}\pbra{\sum_{i,j=1}^{T} \Varx_{\calD_u}\sbra{\abra{\tx{i},\ty{j}}}} = \frac{1}{T^2}\pbra{\Ex_{\calD_u}\sbra{\abra{\wt{\bx},\wt{\by}}^2} - \Ex_{\calD_u}\sbra{\abra{\wt{\bx},\wt{\by}}}^2 } \label{eq:generic-variance}
\end{align}
where $\bx,\by\sim\calD$ are independent samples. In this case, we have by our previous calculation (\Cref{eq:untrunc-means}) that this is in fact equal to
\begin{align}
	\Varx_{\calD_u}[\bM] &= \frac{1}{T^2}\pbra{\Ex_{\calD_u}\sbra{\abra{\wt{\bx}, \wt{\by}}\abra{\wt{\by}, \wt{\bx}}}} \nonumber\\
	&= \frac{1}{T^2}\pbra{\Ex_{{\bx} \sim \calD_u}\sbra{{\wt{\bx}}^{\mathsf{T}}\Ex_{{\by} \sim \calD_u}\sbra{\wt{\by}\cdot \wt{\by}^{\mathsf{T}}}\wt{\bx}}} \nonumber\\
	&= \frac{1}{T^2}\pbra{\Ex_{\calD_u}\sbra{\sum_{1\leq|\alpha|\leq d}\chi_\alpha(\bx)^2}} \nonumber\\
	&= \frac{m}{T^2} \nonumber\\
	&\leq O_d\pbra{\frac{n^d}{T^2}}\label{eq:var-untrunc}
\end{align}
where we used orthonormality as well as the fact that $\Ex_{\calD_u}\sbra{\wt{\by}\cdot \wt{\by}^{\mathsf{T}}} = \mathrm{Id}_m.$
}

{
Turning to the variance, we have 
\begin{align}
	\Varx_{\calD_u}[\bM] 
	= \frac{1}{T^4}\cdot\Varx_{\calD_u}\sbra{\sum_{i,j=1}^{T} \abra{\tx{i},\ty{j}}} 
	= \frac{1}{T^4}\cdot\sum_{i,j,k,\ell = 1}^T \Cov\pbra{\langle \tx{i},\ty{j}\rangle, \langle \tx{k},\ty{\ell} \rangle } \label{eq:generic-variance}
\end{align}
where $\Cov(\bX, \bY) = \E[\bX\cdot\bY] - \E[\bX]\cdot\E[\bY]$ is the covariance of the random variables $\bX$ and $\bY$. 
In particular, 
\begin{align} 
	\Cov\pbra{\langle \tx{i},\ty{j}\rangle, \langle \tx{k},\ty{\ell} \rangle } 
	&= \Ex_{\calD_u}\sbra{\langle \tx{i},\ty{j}\rangle\cdot\langle \tx{k} ,\ty{\ell} \rangle} - \Ex_{\calD_u}\sbra{\langle \tx{i},\ty{j}\rangle}\cdot\Ex_{\calD_u}\sbra{\langle \tx{k} ,\ty{\ell} \rangle} \nonumber \\
	&= \Ex_{\calD_u}\sbra{\langle \tx{i},\ty{j}\rangle\cdot\langle \tx{k} ,\ty{\ell} \rangle} \label{eq:broadsheet}
\end{align}
thanks to~\Cref{eq:untrunc-means}. 
We consider three cases: 
\begin{itemize}
	\item If $i \neq k$ and $j\neq \ell$, then we immediately have $\Cov\pbra{\langle \tx{i},\ty{j}\rangle, \langle \tx{k},\ty{\ell} \rangle } = 0$ thanks to independence.  
	\item If $i = k$ but $j \neq \ell$, then 
	\begin{align*}
		\Ex\sbra{\langle \tx{i},\ty{j}\rangle\cdot\langle \tx{k} ,\ty{\ell} \rangle}
		&= \Ex\sbra{\langle \ty{j}, \tx{i}\rangle\cdot\langle \tx{i} ,\ty{\ell} \rangle} \\ 
		&= \Ex\sbra{\ty{j}\cdot\underbrace{\Ex\sbra{\tx{i}\cdot(\tx{i})^\top}}_{= \mathrm{Id}_m }\cdot\, \ty{\ell}} \\
		&= \Ex\sbra{\abra{\ty{j}, \ty{\ell}}} \\
		&= 0 
	\end{align*} 
	thanks to~\Cref{eq:untrunc-means}. 
	Note that the case when $i \neq k$ but $j = \ell$ is identical by symmetry. 
	\item If $i = k$ and $j = \ell$, then proceeding as in the previous case, we have 
	\begin{align*}
		\Ex\sbra{\langle \tx{i},\ty{j}\rangle\cdot\langle \tx{k} ,\ty{\ell} \rangle}
		 &= \Ex\sbra{\ty{j} \cdot \Ex\sbra{\tx{i}\cdot(\tx{i})^\top}\cdot\ty{j} } \\
		 &= \Ex\sbra{\abra{\ty{j}, \ty{j}} } \\
		 &= \sum_{1 \leq |\alpha| \leq d} \Ex\sbra{\chi_\alpha(\by)^2} = m
	\end{align*}
	thanks to the orthonormality of $\{\chi_\alpha\}$.
\end{itemize}
Combining the above with~\Cref{eq:generic-variance,eq:broadsheet}, we get that 
\begin{equation} \label{eq:var-untrunc}
	\Varx_{\calD_u}[\bM] = \frac{\#\{(i, j, k, \ell) : i = k, j = \ell \}\cdot m}{T^4} = \frac{m}{T^2}. 
\end{equation}

}

\paragraph{Case 2: $\dtv\pbra{\mun, \calD}\geq\eps$ with $\calD\in\Pptf(d, 2^{-n^{1/2(2d+1)}})$.} For brevity, we denote the truncated distribution by 
\[\calD_t := \mun|_{f^{-1}(1)}.\] 
In this case, 
$f:\R^n\to\zo$ is a degree-$d$ PTF with 
\[\vol(f) \in \sbra{2^{-\sqrt{n}}, 1-\epsilon}.\]
We may assume that 
\[f(x) = \Indicator\{p(x) \geq \theta \}\]
where $p:\R^n\to\R$ is a degree-$d$ polynomial with $\widehat{p}(0^n)=0$ and $\|p\|_2^2 = \Var[p] = \sum_\alpha \wh{p}(\alpha)^2 = 1$. 

We have the following easy relation between the Fourier coefficients of $f$ and the means of the characters $\{\chi_\alpha\}$ under the truncated distribution $\calD_t$:
\begin{equation} \label{eq:scaled-fourier-coeffs}
	\Ex_{\bx\sim\calD_t}\sbra{\chi_\alpha(\bx)} = \frac{1}{\vol(f)}\Ex_{\bx\sim\calD_u}\sbra{f(\bx)\chi_\alpha(\bx)} = \frac{\wh{f}(\alpha)}{\vol(f)}.
\end{equation}
We thus have
\begin{align}
	\Ex_{\calD_t}\sbra{ \bM } &= \sum_{1\leq|\alpha|\leq d} \Ex_{\calD_t}\sbra{\chi_\alpha(\bx)}\cdot\Ex_{\calD_t}\sbra{\chi_\alpha(\by)} \nonumber \\
	& = \frac{1}{\vol(f)^2}\sum_{1\leq|\alpha|\leq d} \wh{f}(\alpha)^2 \nonumber\\
	& \geq \frac{1}{\vol(f)^2}\cdot \Omega\pbra{\min\cbra{\vol(f), 1-\vol(f), c^{\Theta(d)}}}^2 \nonumber\\
	& \geq \Omega\pbra{\min\cbra{1, \pbra{\frac{1-\vol(f)}{\vol(f)}}, \frac{c^{\Theta(d)}}{\vol(f)}}}^2 \label{eq:trunc-means}
\end{align}
where the second line follows from \Cref{eq:scaled-fourier-coeffs} and the final inequality is due to \Cref{lem:low-degree-weight}; here $c := c(\mu)$ is as in \Cref{prop:anti-conc}.

{
Turning to the variance of the estimator, we have as before (\Cref{eq:generic-variance}) that 
\begin{equation} \label{eq:reusing-generic-variance-ub}
	\Varx_{\calD_t}[\bM] 
	= \frac{1}{T^4}\cdot\sum_{i,j,k,\ell = 1}^T \Cov\pbra{\langle \tx{i},\ty{j}\rangle, \langle \tx{k},\ty{\ell} \rangle }.
\end{equation}
Note that 
\begin{align} 
	\Cov\pbra{\langle \tx{i},\ty{j}\rangle, \langle \tx{k},\ty{\ell} \rangle } 
	&= \Ex_{\calD_t}\sbra{\langle \tx{i},\ty{j}\rangle\cdot\langle \tx{k} ,\ty{\ell} \rangle} - \Ex_{\calD_t}\sbra{\langle \tx{i},\ty{j}\rangle}\cdot\Ex_{\calD_t}\sbra{\langle \tx{k} ,\ty{\ell} \rangle} \label{eq:muck} \\
	&= \Ex_{\calD_t}\sbra{\langle \tx{i},\ty{j}\rangle\cdot\langle \tx{k} ,\ty{\ell} \rangle} - \pbra{\frac{1}{\vol(f)^2}\sum_{1 \leq |\alpha| \leq d}\wh{f}(\alpha)^2}^2\label{eq:99-hancock-st} \\
	&\leq \Ex_{\calD_t}\sbra{\langle \tx{i},\ty{j}\rangle\cdot\langle \tx{k} ,\ty{\ell} \rangle} \label{eq:dali}
\end{align}
where~\Cref{eq:99-hancock-st} relies on~\Cref{eq:scaled-fourier-coeffs} and the independence of $\tx{i}$ and $\ty{j}$. 
As before, we will consider three cases: 
\begin{itemize}

	\item If $i \neq k$ and $j\neq \ell$, then we immediately have $\Cov\pbra{\langle \tx{i},\ty{j}\rangle, \langle \tx{k},\ty{\ell} \rangle } = 0$ thanks to independence from \Cref{eq:muck}. 

	\item If $i = k$ but $j \neq \ell$, then  
	\begin{align}
		\Ex\sbra{\langle \tx{i},\ty{j}\rangle\cdot\langle \tx{k} ,\ty{\ell} \rangle}
		&= \Ex\sbra{\langle \ty{j}, \tx{i}\rangle\cdot\langle \tx{i} ,\ty{\ell} \rangle} \nonumber \\ 
		&= \Ex\sbra{\ty{j}\cdot\underbrace{\Ex\sbra{\tx{i}\cdot(\tx{i})^\top}}_{=: A}\cdot\, \ty{\ell}} \nonumber \\
		&= \Ex\sbra{\ty{j}}\cdot A\Ex\sbra{\ty{\ell}} \nonumber \\
		&\leq \vabs{\Ex\sbra{\ty{j}}}_2^2\cdot \|A\|_{2} \label{eq:nitro-shandy} \\
		&= \pbra{\frac{1}{\vol(f)^2}\sum_{1\leq |\alpha| \leq d} \wh{f}(\alpha)^2}\cdot\|A\|_{2}\label{eq:purple-haze}
	\end{align} 
	where~\Cref{eq:nitro-shandy} relies on the fact that \smash{$\Ex\sbra{\ty{j}} = \Ex\sbra{\ty{\ell}}$} 
	and \Cref{eq:purple-haze} follows from \Cref{eq:scaled-fourier-coeffs}. 
	Also, note that $A$ is an $m\times m$ matrix whose $(\beta, \gamma)^\text{th}$ entry (for $\beta, \gamma \in \N^n$ satisfying $1 \leq |\beta|, |\gamma| \leq d$) is given by 
	\begin{equation} \label{eq:a-def}
		A_{\beta, \gamma} = \Ex_{\bx\sim\calD_t}\sbra{\chi_\beta(\bx)\chi_{\gamma}(\bx)}.
	\end{equation}
	We will control the final quantity in~\Cref{eq:purple-haze} shortly; for now, note that the case when $i \neq k$ but $j = \ell$ is identical by symmetry.

	\item If $i = k$ and $j = \ell$, then 
	\begin{align}
		\Ex_{\calD_t}\sbra{\langle \tx{i},\ty{j}\rangle\cdot\langle \tx{k} ,\ty{\ell} \rangle} &= \Ex_{\calD_t}\sbra{\langle \tx{i},\ty{j}\rangle^2} \nonumber \\
		&= \Ex_{\calD_t}\sbra{\sum_{1\leq|\beta|,|\gamma|\leq d} \chi_{\beta}(\bx)\chi_{\gamma}\pbra{\bx}\chi_{\beta}(\by)\chi_{\gamma}\pbra{\by} } \nonumber \\
		&= \sum_{1\leq |\beta|, |\gamma| \leq d} \Ex_{\calD_t}\sbra{\chi_\beta(\bx)\chi_\gamma(\bx)}^2. \label{eq:trunc-case-3-final-quantity}
	\end{align}
\end{itemize}
Note that a similar term, namely $\Ex_{\calD_t}\sbra{\chi_\beta(\bx)\chi_\gamma(\bx)}$, appears in both~\Cref{eq:a-def,eq:trunc-case-3-final-quantity}. We will rely on the level-$k$ inequality (\Cref{prop:level-k-inequalities}) to control this quantity. 

As in \Cref{eq:scaled-fourier-coeffs}, we have
\begin{align}
	\sum_{1\leq |\beta|,|\gamma|\leq d}	\Ex_{\calD_t}\sbra{\chi_{\beta}\chi_{\gamma}}^2 &= \sum_{1\leq |\beta|,|\gamma|\leq d}	\pbra{\frac{1}{\vol(f)}\Ex_{\calD_u}\sbra{f\cdot\chi_{\beta}\chi_{\gamma}}}^2 \nonumber\\
	&= \frac{1}{\vol(f)^2}\sum_{{1 \leq}|\beta|,|\gamma|\leq d} \pbra{\Ex_{\calD_u}\sbra{f\cdot\pbra{\sum_{|\alpha|\leq 2d} \abra{\chi_\beta\chi_\gamma, \chi_\alpha}\chi_\alpha }}}^2 \nonumber \\
	&= \frac{1}{\vol(f)^2}\sum_{1\leq|\beta|,|\gamma|\leq d} \pbra{\sum_{|\alpha|\leq 2d} \abra{\chi_\beta\chi_\gamma, \chi_\alpha}\wh{f}(\alpha)}^2 \nonumber\\
	&\leq \frac{O_d(1)}{\vol(f)^2}\sum_{1\leq |\beta|,|\gamma|\leq d}\pbra{\sum_{|\alpha|\leq 2d} {|}\wh{f}(\alpha){|}\cdot\mathbf{1}\pbra{\abra{\chi_\beta\chi_\gamma, \chi_\alpha}\neq 0}}^2 \label{eq:lin-coeffs-small}
\end{align}
where the $O_d(1)$ factor in \Cref{eq:lin-coeffs-small} comes from the RHS of 
 \Cref{lem:ub-linearization-coeffs}. Now, observe that for a fixed $(\beta, \gamma)$ pair with $1 \leq |\beta|,|\gamma| \leq d$ there are $O_d(1)$ many multi-indices $\alpha$ such that $\abra{\chi_\beta\chi_\gamma, \chi_\alpha} \neq 0$ as we must have $\supp(\alpha) \sse \supp(\beta) \cup \supp(\gamma)$. 
{Combining this observation with an easy application of the Cauchy-Schwarz inequality:
\[\pbra{\sum_{i=1}^{t} a_i}^2 \leq t\cdot\pbra{\sum_{i=1}^{t} a_i^2},\]
we have}
\begin{align*}
	\sum_{1\leq |\beta|,|\gamma|\leq d}	\Ex_{\calD_t}\sbra{\chi_{\beta}\chi_{\gamma}}^2
	&\leq\frac{O_d(1)}{\vol(f)^2}\sum_{1\leq |\beta|,|\gamma|\leq d}\sum_{|\alpha|\leq 2d}  \wh{f}(\alpha)^2\cdot\mathbf{1}\pbra{\abra{\chi_\beta\chi_\gamma, \chi_\alpha}\neq 0} \\
	&= \frac{O_d(1)}{\vol(f)^2}\sum_{|\alpha|\leq 2d}\sum_{1\leq |\beta|,|\gamma|\leq d}  \wh{f}(\alpha)^2\cdot\mathbf{1}\pbra{\abra{\chi_\beta\chi_\gamma, \chi_\alpha}\neq 0}\\
	&=\frac{O_d(1)}{\vol(f)^2}\sum_{k=0}^{2d} \sum_{|\alpha|= k}\sum_{1\leq |\beta|,|\gamma|\leq d}  \wh{f}(\alpha)^2\cdot\mathbf{1}\pbra{\abra{\chi_\beta\chi_\gamma, \chi_\alpha}\neq 0} \nonumber \\
	&\leq \frac{O_d(1)}{\vol(f)^2}\sum_{k=0}^{2d} \sum_{|\alpha|= k} \wh{f}(\alpha)^2\cdot O_d\pbra{n^{d-\ceil{k/2}}} \nonumber 
\end{align*}
where the final inequality is due to \Cref{lem:linearization}. We can rephrase the last inequality as
\begin{align}
	\sum_{1\leq |\beta|,|\gamma|\leq d}	\Ex_{\calD_t}\sbra{\chi_{\beta}\chi_{\gamma}}^2 &\leq\frac{O_d(1)}{\vol(f)^2}\sum_{k=0}^{2d} O_d\pbra{n^{d-\ceil{k/2}}}\cdot \bW^{=k}[f]. \label{eq:bomboloni}
\end{align}

Returning to the variance calculation, combining~\Cref{eq:reusing-generic-variance-ub,eq:dali,eq:purple-haze,eq:trunc-case-3-final-quantity} yields the following:
\begin{align} 
	\Varx_{\calD_t}[\bM] 
	&\leq \frac{1}{T^4}\cdot\pbra{T^3\pbra{\frac{\bW^{\leq d}[f]\cdot \|A\|_{2}}{\vol(f)^2}} + T^2\pbra{\sum_{1\leq |\beta|,|\gamma|\leq d}	\Ex_{\calD_t}\sbra{\chi_{\beta}\chi_{\gamma}}^2}}. 
	\label{eq:chobani-coffee} 
\end{align}

As $f:\R^n\to\zo$ is Boolean-valued, we have 
\begin{equation} \label{eq:naive-fourier-ub}
	\vol(f) = \Ex_{\bx\sim\mun}[f(\bx)] = \Ex_{\bx\sim\mun}[f(\bx)^2] = \sum_{\alpha\in\N^n} \wh{f}(\alpha)^2
	\qquad\text{and so}\qquad 
	\bW^{\leq k}[f] \leq \vol(f)
\end{equation}
for all $k \in \N$. 
If $\vol(f) \geq 2^{-d}$, then~\Cref{eq:bomboloni,eq:naive-fourier-ub} imply that 
\[
	\sum_{1\leq |\beta|,|\gamma|\leq d}	\Ex_{\calD_t}\sbra{\chi_{\beta}\chi_{\gamma}}^2
	 \leq \frac{O_d(n^d)}{\vol(f)}.
\] 
Furthermore, recall that 
\[
	\|A\|_{2} \leq \|A\|_F = \sqrt{\sum_{1\leq |\beta|,|\gamma| \leq d} \Ex_{\calD_t}\sbra{\chi_{\beta}\chi_{\gamma}}^2} = \sqrt{\frac{O_d(n^d)}{\vol(f)}}
\]
thanks to the definition of $A$ (cf. \Cref{eq:a-def}).
Combining these observations with~\Cref{eq:chobani-coffee} and applying~\Cref{eq:naive-fourier-ub} gives 
\[
	\Varx_{\calD_t}[\bM] \leq O_d(1)\cdot\pbra{\frac{n^{d/2}}{T\cdot \vol(f)^{3/2}} + \frac{n^d}{T^2\cdot\vol(f)}} = O_d\pbra{\frac{n^d}{T^2}}. 
\]
as $\vol(f) \geq 2^{-d}$. 

When $\vol(f) \leq 2^{-d}$, the bound due to~\Cref{eq:naive-fourier-ub} is too loose for our purposes and we will instead rely on the level-$k$ inequalities (\Cref{prop:level-k-inequalities}). 
We will also require a bit more care when bounding $\|A\|_{2}$ and cannot directly pass to the Frobenius norm. 
Recalling~\Cref{eq:bomboloni}, 
\begin{align}
	\sum_{1\leq |\beta|,|\gamma|\leq d}	\Ex_{\calD_t}\sbra{\chi_{\beta}\chi_{\gamma}}^2 
	&\leq\frac{O_d(1)}{\vol(f)^2}\sum_{k=0}^{2d} O_d\pbra{n^{d-\ceil{k/2}}}\cdot \bW^{=k}[f]\nonumber \\
	&\leq \frac{O_d(1)}{\vol(f)^2}\sum_{k=0}^{2d} n^{d-\ceil{k/2}}\cdot \pbra{\vol(f)^2\cdot \pbra{\log\pbra{\frac{1}{\vol(f)}}}^k} \label{eq:drive}\\
	&= O_d(1)\pbra{\sum_{k=0}^{2d} n^{d-\ceil{k/2}}\cdot \pbra{\log\pbra{\frac{1}{\vol(f)}}}^k},  \label{eq:first-level-k-app} \\
	&\leq O_d(1)\cdot n^{d}, \label{eq:plow}
\end{align}
where \Cref{eq:drive} is by the level-$k$ inequality (\Cref{prop:level-k-inequalities}),\footnote{Note that we can indeed use the level-$k$ inequalities with $k = 2d$ because $\vol(f) \leq 2^{-d}$.} and~\Cref{eq:plow} follows from our assumption that $\vol(f) \geq 2^{-n^{1/2(2d+1)}}$.

Recalling~\Cref{eq:purple-haze}, we have 
\[
	\|A\|_2 \leq \vabs{\mathrm{diag}(A)}_2 + \vabs{A - \mathrm{diag}(A)}_2 \leq \vabs{\mathrm{diag}(A)}_2 + \vabs{A - \mathrm{diag}(A)}_F
\]
where $\mathrm{diag}(A)$ is the $m\times m$ diagonal matrix given by the diagonal entries of $A$. 

We will first give an upper bound on $\|\mathrm{diag}(A)\|_2$. Note that 
\begin{align*}
	\mathrm{diag}(A)_{\beta, \beta} = \Ex_{\calD_t}\sbra{\chi_\beta^2} &= \frac{1}{\vol(f)}\Ex_{\calD_u}\sbra{f\cdot\chi_\beta^2} \\
	&= \frac{1}{\vol(f)}\sum_{|\alpha|\leq 2d}\Ex_{\calD_u}\sbra{f\cdot\chi_\alpha \abra{\chi_\beta^2, \chi_\alpha}} \\
	&= \frac{1}{\vol(f)} \sum_{|\alpha|\leq 2d} \wh{f}(\alpha)  \abra{\chi_\beta^2, \chi_\alpha} \\
	&\leq \frac{1}{\vol(f)} \sqrt{\bW^{\leq 2d}[f]}\cdot\sqrt{\sum_{|\alpha|\leq 2d} \wh{\chi^2_\beta}(\alpha)^2} \tag{Cauchy-Schwarz} \\
	&\leq O_d(1) \pbra{\log\pbra{\frac{1}{\vol(f)}}}^{d} \cdot \sqrt{\abra{\chi_\beta^2, \chi_\beta^2}} \tag{\Cref{prop:level-k-inequalities}} \\
	&\leq O_d(1)\pbra{\log\pbra{\frac{1}{\vol(f)}}}^d. \tag{\Cref{lem:ub-linearization-coeffs}}
\end{align*}
It immediately follows that 
\begin{equation}
\label{eq:diag-op-ub}
	\vabs{\mathrm{diag}(A)}_2 \leq O_d(1)\pbra{\log\pbra{\frac{1}{\vol(f)}}}^d. 
\end{equation}

Turning to $\vabs{A - \mathrm{diag}(A)}_F$ and proceeding as we did to obtain~\Cref{eq:lin-coeffs-small}, we get
\begin{align*}
	\vabs{A - \mathrm{diag}(A)}_F ^2= \sum_{\substack{1 \leq |\beta|, |\gamma| \leq d \\ \beta \neq \gamma}} \Ex_{\calD_t}\sbra{\chi_\beta\chi_\gamma}^2 
	&\leq \frac{O_d(1)}{\vol(f)^2}\sum_{\substack{1\leq |\beta|,|\gamma|\leq d \\ \beta\neq\gamma}}\pbra{\sum_{|\alpha|\leq 2d} {|}\wh{f}(\alpha){|}\cdot\mathbf{1}\pbra{\abra{\chi_\beta\chi_\gamma, \chi_\alpha}\neq 0}}^2 \\ 
	&\leq \frac{O_d(1)}{\vol(f)^2}\sum_{\substack{1\leq |\beta|,|\gamma|\leq d \\ \beta\neq\gamma}}{\sum_{|\alpha|\leq 2d} \wh{f}(\alpha)^2\cdot\mathbf{1}\pbra{\abra{\chi_\beta\chi_\gamma, \chi_\alpha}\neq 0}} \\
	&= \frac{O_d(1)}{\vol(f)^2}\sum_{k=0}^{2d}\sum_{|\alpha| = k}\sum_{\substack{1\leq |\beta|,|\gamma|\leq d \\ \beta\neq\gamma}}{\wh{f}(\alpha)^2\cdot\mathbf{1}\pbra{\abra{\chi_\beta\chi_\gamma, \chi_\alpha}\neq 0}} \\
	&= \frac{O_d(1)}{\vol(f)^2}\sum_{k=1}^{2d}\sum_{|\alpha| = k}\sum_{\substack{1\leq |\beta|,|\gamma|\leq d \\ \beta\neq\gamma}}{\wh{f}(\alpha)^2\cdot\mathbf{1}\pbra{\abra{\chi_\beta\chi_\gamma, \chi_\alpha}\neq 0}},
\end{align*}
where the second inequality uses Cauchy--Schwarz and the same reasoning that was used immediately after \Cref{eq:lin-coeffs-small}, and the final equality is because if $|\alpha| = 0$ then $\chi_\alpha \equiv 1$, and so 
\[
	\abra{\chi_\beta\chi_\gamma, \chi_\alpha} = \abra{\chi_\beta,\chi_\gamma} = \begin{cases}
 1 & \beta = \gamma \\ 
 0 & \beta\neq \gamma	
 \end{cases}
\]
thanks to orthonormality of the $\{\chi_\beta\}$. Proceeding as we did to obtain~\Cref{eq:first-level-k-app} (via~\Cref{lem:linearization} and~\Cref{prop:level-k-inequalities}) then gives us
\begin{equation} \label{eq:second-level-k-app}
	\vabs{A - \mathrm{diag}(A)}_F^2 \leq O_d(1)\pbra{\sum_{k=1}^{2d} n^{d-\ceil{k/2}}\cdot \pbra{\log\pbra{\frac{1}{\vol(f)}}}^k}.
\end{equation}
Note that this is identical to~\Cref{eq:first-level-k-app}, except that the summation starts at $k=1$; this difference is crucial. 
To summarize, we have 
\[
	\|A\|_2 \leq O_d(1)\pbra{\pbra{\log\pbra{\frac{1}{\vol(f)}}}^d + \sqrt{\sum_{k=1}^{2d} n^{d-\ceil{k/2}}\cdot \pbra{\log\pbra{\frac{1}{\vol(f)}}}^k}}.
\]

Returning to~\Cref{eq:chobani-coffee}, the level-$k$ inequality (\Cref{prop:level-k-inequalities}) and the above bound on $\|A\|_2$ together yield the following:   
\begin{align}
	\frac{\bW^{\leq d}[f]\cdot\|A\|_2}{\vol(f)^2} 
	&\leq O_d\pbra{\pbra{\log\pbra{\frac{1}{\vol(f)}}}^{2d} + \sqrt{\sum_{k=1}^{2d} n^{d -\ceil{k/2}}\cdot\pbra{\log\pbra{\frac{1}{\vol(f)}}}^{2d + k}}}. \nonumber
\end{align}
Thanks to our assumption that $\vol(f) \geq 2^{-n^{1/2(2d+1)}}$, we have 
\[
	\pbra{\log\pbra{\frac{1}{\vol(f)}}}^{2d} \leq n^{1/2}
\]
and furthermore 
\[
	\pbra{\log\pbra{\frac{1}{\vol(f)}}}^{2d+k} \leq n^{\frac{2d+k}{2(2d+1)}} \leq n^{k/2} \qquad \text{for}~k \geq 1. 
\]
In particular, it follows that
\begin{equation} \label{eq:61200-midterm-1}
	\frac{\bW^{\leq d}[f]\cdot\|A\|_2}{\vol(f)^2}  \leq O_d(n^{d/2}). 
\end{equation}

}

Putting \Cref{eq:chobani-coffee,eq:plow,eq:61200-midterm-1} together, we have 
\begin{equation}
	\label{eq:var-trunc}
	\Varx_{\calD_t}[\bM] \leq O_d\pbra{\frac{n^d}{T^2}}.
\end{equation}

\medskip

%

%
%
%
%
%

To summarize, when $\calD = \mun$, we have from \Cref{eq:untrunc-means,eq:var-untrunc} that 
\begin{equation}
\Ex_{\calD_u}[\bM] = 0 \qquad \text{and} \qquad \Varx_{\calD_u}[\bM] = O_d\pbra{\frac{n^d}{T^2}}.
\label{eq:over}
\end{equation}
On the other hand, when $\calD \in \Pptf(d, 2^{-n^{1/2(2d+1)}})$ with $\dtv(\calD, \mun) \geq \epsilon$ , we have from \Cref{eq:trunc-means,eq:var-trunc} that
\begin{equation}
\Ex_{\calD_t}[\bM]\geq \Omega\pbra{\min\cbra{1, \pbra{\frac{\eps}{1-\epsilon}}, \pbra{\frac{c^{\Theta(d)}}{1-\epsilon}}}}^2 \qquad \text{and}\qquad \Varx_{\calD_t}[\bM] \leq O_d\pbra{\frac{n^d}{T^2}}.
\label{eq:bones}
\end{equation}
For the number of samples being
\[T = \Theta_d\pbra{\frac{n^{d/2}}{\min\cbra{1, \eps/(1-\eps), c^{\Theta(d)}/(1-\eps)}^2}},\]
the correctness of the distinguishing algorithm \Cref{alg:ptf} follows directly from \Cref{eq:over} and \Cref{eq:bones} by a simple application of  Chebyshev's inequality.
\end{proofof}


\section{An $\Omega(n^{d/2})$-Sample Lower Bound for Distinguishing PTFs}
\label{sec:lb}

In this section, we show that there exists a hypercontractive product distribution for which distinguishing truncation by a degree-$d$ PTF requires $\Omega(n^{d/2})$ samples. In fact, we prove this for arguably the simplest hypercontractive product measure, namely the uniform distribution on the Boolean hypercube $\bn$. As discussed in the introduction, this strengthens the $\Omega(n^{d/4})$-sample lower bound for distinguishing low-volume PTFs (cf. \Cref{appendix:baseline-distinguisher}) and matches our upper bound from \Cref{thm:ptf-ub} up to an absolute constant factor (depending only on the degree-$d$).

\begin{notation}
Throughout this section, we will write $\calD_u$ to denote the uniform distribution over $\bn$.
\end{notation}

\begin{theorem} \label{thm:lb-ptf}
	Let $A$ be any algorithm which is given access to samples from some unknown distribution $\calD$ over $\R^n$ and has the following performance guarantee:
	\begin{enumerate}
		\item If $\calD = \calD_u$, then with probability at least $99/100$ the algorithm $A$ outputs ``un-truncated''; and 
		\item If $\calD = \calD_u|_{f^{-1}(1)}$ where $f:\bn\to\zo$ is a degree-$d$ PTF with $\vol(f) \in \sbra{1/2 \pm o_n(1)}$, then with probability at least $99/100$ the algorithm $A$ outputs ``truncated''.
	\end{enumerate}
	Then $A$ must use at least $\Omega(n^{d/2})$ samples from $\calD$.
\end{theorem}

After setting up preliminary results in \Cref{subsec:lb-prelims}, we will prove \Cref{thm:lb-ptf} via establishing a total variation distance bound over an appropriate pair of distributions in \Cref{subsec:lb-proof}. 

\subsection{Useful Preliminaries}
\label{subsec:lb-prelims}

We first recall a recent upper bound on the total variation distance between $n$-dimensional Gaussians due to Devroye, Mehrabian and Reddad \cite{DMR20}.

\begin{theorem} 
[Theorem~1.1 of \cite{DMR20}]  \label{thm:DMR20}
Let $\Sigma_1,\Sigma_2$ be two $m \times m$ positive definite covariance matrices and let $\lambda_1,\dots,\lambda_m$ be the eigenvalues of $\Sigma_1^{-1} \Sigma_2 - I_m$. Then
\[
\dtv\pbra{N(0^m,\Sigma_1),N(0^m,\Sigma_2)} \leq {\frac 3 2} \cdot \min\left\{1,\sqrt{\lambda_1^2 + \cdots + \lambda_m^2}\right\} 
\]
\end{theorem}

\begin{figure}

\centering

\fbox{\parbox{0.75\textwidth}{
\begin{enumerate}
	\item Draw the vector $\pbra{\wh{\bp}(S)}_{S \in {[n]\choose d}} =: \wh{\bp} \sim N(0, 1)^{{n\choose d}}$.
	\item Set \[\bp(x) := \sum_{{\substack{S\sse[n]\\|S|=d}}} \wh{\bp}(S)x_S \qquad \text{where}\qquad x_S := \prod_{i\in S}x_i,\]
	and output the degree-$d$ PTF $\boldf(x) := \mathbf{1}\cbra{\bp(x) \geq 0}$.  
\end{enumerate}
}}

\caption{A draw of $\boldf$ from the distribution $\calF_d$}
\label{fig:ptf-dist}

\end{figure}

Next, we introduce a distribution $\calF_d$ (cf. \Cref{fig:ptf-dist}) over multilinear degree-$d$ PTFs over the Boolean hypercube. The following lemma shows that a random PTF $\boldf$ drawn from $\calF_d $ is close to being balanced with high probability.

\begin{lemma} \label{lemma:balanced}
For $\boldf\sim\calF_d$ as described in \Cref{fig:ptf-dist}, we have with probability at least $1-o(1)$ that 
\[\vol(\boldf) \in \left[\frac{1}{2} - o(1), \frac{1}{2} + o(1)\right].\]
\end{lemma}

The proof of \Cref{lemma:balanced} will make use of the following well-known fact characterizing the probability that two vectors are both satisfying assignments of a random origin-centered LTF that has Gaussian coordinates. The fact follows immediately from spherical symmetry of the standard Gaussian distribution, and is sometimes known as ``Sheppard's formula'' (see Section 11.1 of \cite{ODonnell2014}).

\begin{fact} \label{fact:gaussian}
Let $u,v$ be unit vectors in $\R^m$.
We have
\[
\Prx_{\bg \sim N(0,I_m)}[\bg \cdot u \geq 0 \text{~and~}
\bg \cdot v \geq 0]
=
{\frac 1 2} - {\frac {\arccos(u \cdot v)}{2 \pi}}.
\]
\end{fact}

We now turn to the proof of \Cref{lemma:balanced}.

\begin{proofof}{\Cref{lemma:balanced}}
If $d$ is odd then $\bp(-x)=-\bp(x)$ (i.e. $\bp$ is an odd function); since for any fixed $x$ we have that $\bp(x)=0$ with probability 0, it follows that $\vol(\boldf)=1/2$ with probability $1$ for odd $d$.  So in the rest of the argument we may suppose $d$ is even (though the argument we give will work for all $d \geq 1$, even or odd).

We first observe that $\E[\vol(\boldf)]=1/2$: this is because
\begin{align*}
\E[\vol(\boldf)]&=
\Ex_{\bp}\left[
\Prx_{\bx \sim \bn}[\bp(\bx)\geq 0]\right]
=
\Ex_{\bx}\left[
\Prx_{\bp}\left[
\bp(\bx) \geq 0
\right]\right] = 1/2,
\end{align*}
where the last equality is because for each $x \in (\R^n \setminus 0^n)$ we have $\Prx_{\bp}[\bp(x) \geq 0]=1/2.$

Next, we upper bound the variance of $\vol(\boldf)$. We have
\[
\Var[\vol(\boldf)] =
\Ex[\vol(\boldf)^2] - \Ex[\vol(\boldf)]^2
=
\Ex[\vol(\boldf)^2] - {\frac 1 4},
\]
so our goal is to show that $\Ex[\vol(\boldf)^2]$ is $1/4 + o(1)$.
\begin{align}
\Ex[\vol(\boldf)^2] &=
\Ex_{\bp}\left[
\Prx_{\bx \sim \bn}[\bp(\bx)\geq 0]^2
\right]\nonumber \\
&=
\Ex_{\bp}\left[
\Prx_{\bx,\by \sim \bn}[\bp(\bx)\geq 0 \text{~and~}\bp(\by)\geq 0]
\right]\nonumber \\
&=
\Ex_{\bx,\by \sim \bn}\left[
\Prx_{\bp}[\bp(\bx) \geq 0 \text{~and~}\bp(\by)\geq0]
\right]\nonumber \\
&=
\Ex_{\bx,\by \sim \bn}\left[
\Prx_{\bp}\left[\sum_{S \in {[n] \choose d}} \widehat{\bp}(S) \bx_S \geq 0 \text{~and~}
\sum_{S \in {[n] \choose d}} \widehat{\bp}(S) \by_S \geq 0\right]
\right]\nonumber \\
&=
\Ex_{\bx,\by \sim \bn}\left[
\Prx_{\widehat{\bp} \sim N(0,1)^{[n] \choose d}}
\left[
\widehat{\bp} \cdot \Phi(\bx) \geq 0 \text{~and~}
\widehat{\bp} \cdot \Phi(\by) \geq 0
\right]
\right] \label{eq:rutabaga}
\end{align}
where we define $\Phi: \bn\to\bits^{n\choose d}$ as
\[\Phi(x) = \pbra{x_S}_{S\in{[n]\choose d}}.\]
We write $\Phi'(x)$ to denote ${\Phi(x)}/{\sqrt{{n \choose d}}}$, so $\Phi'(x)$ is $\Phi(x)$ normalized to be a unit vector in $\R^{{n \choose d}}.$ By \Cref{fact:gaussian} we have that
\begin{align}
(\ref{eq:rutabaga}) &=
\Ex_{\bx,\by \sim \bn}\left[
{\frac 1 2} - {\frac {\arccos(\Phi'(x) \cdot \Phi'(y))}{2 \pi}}
\right]. \label{eq:blue-pumpkin}
\end{align}

Let
\[
h(x,y) := \Phi'(x) \cdot \Phi'(y) = {\frac {\sum_{S \in {[n] \choose d}} x_S y_S}{{n \choose d}}}.
\]
We see that $h(x,y)$ is a degree-$2d$ homogeneous polynomial with ${n \choose d}$ Fourier coefficients each of which is ${n \choose d}^{-1}$, so $\E_{\bx,\by \sim \bn}[h(\bx,\by)]=0$ and {$\Var[h]=\|h\|_2^2={n \choose d} \cdot {n \choose d}^{-2} = {n \choose d}^{-1}$}.
It follows from Chebyshev's inequality that
\[
\Pr\left[h(x,y) > t / {n \choose d}^{1/2}\right] \leq 1/t^2,
\]
so in particular we have that $\Pr_{\bx,\by \sim \bn}[\Phi'(\bx) \cdot \Phi'(\by) > {n \choose d}^{-1/3}] \leq {n \choose d}^{-1/3}.$
Since ${\frac 1 2} - {\frac {\arccos(\Phi'(x) \cdot \Phi'(y))}{2 \pi}}$ is always at most $1/2$, we get that 
\begin{align*}
(\ref{eq:blue-pumpkin}) &\leq 
\left({\frac 1 2} - {\frac {\arccos({n \choose d}^{-1/3})}{2\pi}}\right) 
+
{\frac 1 2} \cdot \Prx_{\bx,\by \sim \bn}\left[\Phi'(\bx) \cdot \Phi'(\by) > {n \choose d}^{-1/3}\right]\nonumber\\
&\leq \left({\frac 1 2} - {\frac {\arccos({n \choose d}^{-1/3})}{2\pi}}\right) 
+ {\frac 1 2}\cdot {n \choose d}^{-1/3}\nonumber\\
&\leq {\frac 1 2} - {\frac {\pi/2 - 2 {n \choose d}^{-1/3}}{2 \pi}} + {\frac 1 2}\cdot {n \choose d}^{-1/3}\nonumber\\
&\leq {\frac 1 4} + C {n \choose d}^{-1/3}
\end{align*}
for an absolute constant $C>0$.
So $\Var[\vol(\boldf)] \leq C{n \choose d}^{-1/3}$, and hence Chebyshev's inequality tells us that $|\vol(\boldf) - {\frac 1 2}| \geq {n \choose d}^{-1/12}$ with probability at most $O({n \choose d}^{-1/6})=o(1),$ as claimed.
\end{proofof}

\subsection{Proof of \Cref{thm:lb-ptf}}
\label{subsec:lb-proof}

Given \Cref{lemma:balanced}, to prove \Cref{thm:lb-ptf}, it suffices to prove an $\Omega(n^{d/2})$ lower bound on the sample complexity of any algorithm $A$ that outputs ``un-truncated'' with probability at least $9/10$ if it is given samples from $\calD_u$  and outputs ``truncated'' with probability at least $91/100$ if it is given samples from $\calD_u|_{\boldf^{-1}(1)}$ where $\boldf \sim \calF_d$. We do this in the rest of this section and start by introducing the ``un-truncated'' and ``truncated'' distributions that we will show are hard to distinguish.

We will consider two different distributions ${\cal D}_1$ and ${\cal D}_2$, each of which is a distribution over sequences of $m = c {n \choose d}^{1/2}$ many points in $\bn$ where $c$ is an absolute constant that we will set later. These distributions over sequences are defined as follows:

\begin{enumerate}

\item A draw of $\overline{\bx}=(\bx^{(1)},\dots,\bx^{(m)})$ from ${\cal D}_1$  is obtained by having each $\bx^{(i)} \in \R^n$ be drawn independently and uniformly from $\calD_u$.

\item A draw of $\overline{\bx}=(\bx^{(1)},\dots,\bx^{(m)})$ from ${\cal D}_2$ is obtained by drawing (once and for all) a function $\boldf \sim \calF_d$ as described in \Cref{fig:ptf-dist} and having each $\bx^{(i)}$ be distributed independently according to $\calD_u|_{\boldf}$, i.e. each $\bx^{(i)}$ is drawn from $\calD_u$ conditioned on  satisfying $\bp(x^{(i)}) \geq 0$.

\end{enumerate}

We further define two more distributions ${\cal D}'_1$ and ${\cal D}'_2$, each of which is a distribution over sequences of $3m$ points in $\bn \cup \{0^n\}$:

\begin{enumerate}

\item A draw of $\overline{\bx}=(\bx^{(1)},\dots,\bx^{(3m)})$ from ${\cal D}'_1$ is obtained as follows: $\bx^{(i)}$ is taken to be $\bb_i \cdot \bu^{(i)}$ where each $\bb_i$ is an independent uniform draw from $\zo$ and each $\bu^{(i)}$ is an independent draw from $\calD_u$.

\item A draw of $\overline{\bz}=(\bz{(1)},\dots,\bz^{(3m)})$ from ${\cal D}'_2$ is obtained by drawing (once and for all) a function $\boldf \sim \calF_d$ as described in \Cref{fig:ptf-dist} and taking $\bz^{(i)}$ to be $\boldf(\bu^{(i)})\cdot \bu^{(i)} $, where each $\bu^{(i)}$ is an independent draw from $\calD_u$.

\end{enumerate}

The following claim shows that distinguishing $\calD_1'$ and $\calD_2'$ is no harder than distinguishing $\calD_1$ and $\calD_2$. 

\begin{claim} \label{claim:lb-reduction}
	Suppose there exists an algorithm $A$ that successfully distinguishes between ${\cal D}_1$ and ${\cal D}_2$ (i.e. it outputs ``un-truncated'' with probability at least 9/10 when run on a draw from ${\cal D}_1$ and outputs ``truncated'' with probability at least 91/100 when run on a draw from ${\cal D}_2$).
	
Then there exists an algorithm $A'$ which successfully distinguishes between ${\cal D}'_1$ and ${\cal D}'_2$ (i.e. it outputs ``un-truncated'' with probability at least 89/100 when run on a draw from ${\cal D}'_1$ and outputs ``truncated'' with probability at least 90/100 when run on a draw from ${\cal D}'_2$).
\end{claim}

\begin{proof}
Let $A'$ be the algorithm that simply takes the first $m$ non-$0^n$ points in its $3m$-point input sequence and uses them as input to $A$.  We may suppose that $A'$ fails if there are fewer than $m$ non-$0^n$ points in the $3m$-point sample: It is easy to see that whether the input to $A'$ is drawn from ${\cal D}'_1$ or ${\cal D}'_2$, the probability of failure is at most $o(1)$ and hence is negligible; this uses \Cref{lemma:balanced}, i.e.~the fact that in the no-case, with very high probability $\boldf$ is roughly balanced. 

We then have that 
\begin{itemize}
	\item If the input sequence is a draw from ${\cal D}'_1$ then the input that $A'$ gives to $A$ is distributed (up to statistical distance $O(m/2^n)$) according to ${\cal D}_1$. This statistical distance comes because $A'$ never passes the $0^n$-vector along as an input to $A$, whereas under the distribution ${\cal D}_1$, each of the $m$ points has a $1/2^n$ chance of being $0^n$; and
	\item If the input is a draw from ${\cal D}'_2$ then the input sequence that $A'$ gives to $A$ is distributed (up to statistical distance $o(1)$) according to ${\cal D}_2$. This statistical distance comes from the $o(1)$ chance that $\boldf$ has $\vol[\boldf] \notin [1/2 - o(1),1/2 + o(1)]$, plus the fact that even if $\vol[\boldf] \in [1/2 - o(1),1/2 + o(1)]$, algorithm  $A'$ never passes along the 0-vector as an input to $A$.
\end{itemize}
This completes the proof of the claim. 
\end{proof}

Given \Cref{claim:lb-reduction}, to prove \Cref{thm:lb-ptf} it suffices to show that any algorithm $A'$ with the performance guarantee described above must have $3m = \Omega(n^{d/2})$.  This is a consequence of the following statement:
\begin{equation}
\label{eq:dtvprimes}
\text{If $m = c {n \choose d}^{1/2}$, then~}\dtv({\cal D}'_1,{\cal D}'_2) \leq 0.01.
\end{equation}
The rest of the proof establishes (\ref{eq:dtvprimes}).

Consider a coupling of the two distributions ${\cal D}'_1$ and ${\cal D}'_2$ as follows. A draw from the joint coupled distribution of $({\cal D}'_1,{\cal D}'_2)$ is generated in the following way:

\begin{itemize}

\item Let $(\bb_1,\dots,\bb_{3m})$ be a uniform random string from $\zo^{3m}$, let $\boldf \sim \calF_d$, and let $\bu^{(1)},\dots,\bu^{(3m)}$ be $3m$ independent draws from $\calD_u$.
\item The draw from the joint coupled distribution is $(\overline{\bx},\overline{\bz})$ where
\[
\text{for each $i =1,\dots,3m$,~~~}\bx^{(i)} = \bb_i \bu^{(i)} \quad \text{and}\quad
 \bz^{(i)}=\boldf(\bu^{(i)}) \bu^{(i)}.
\]
\end{itemize}
It is easily verified that this is indeed a valid coupling of ${\cal D}'_1$ and ${\cal D}'_2$.
Writing $\overline{\bu}$ to denote $(\bu^{(1)},\dots,\bu^{(m)})$, we have that
\begin{align}
\dtv({\cal D}'_1,{\cal D}'_2) &= \dtv((\bx^1,\dots,\bx^{3m}),(\bz^1,\dots,\bz^{3m})) \nonumber\\
&=\dtv\left((\bb_1\bu^{(1)},\dots,\bb_{3m}\bu^{(3m)}),
\left(
\boldf(\bu^{(1)}) \bu^{(1)},
\dots,
\boldf(\bu^{(3m)}) \bu^{(3m)}
 \right)\right) \nonumber\\
 &\leq \Ex_{\overline{\bu}} \left[\dtv\left((\bb_1\bu^{(1)},\dots,\bb_{3m}\bu^{(3m)}),
\left(
\boldf(\bu^{(1)}) \bu^{(1)},
\dots,
\boldf(\bu^{(3m)}) \bu^{(3m)}
 \right)\right)\right] \label{eq:gourd}\\
 &= \Ex_{\overline{\bu}}\left[\dtv\left((\bb_1,\dots,\bb_{3m}),
 \left(
\boldf(\bu^{(1)}),
\dots,
\boldf(\bu^{(3m)})
\right)\right)\right],\label{eq:pumpkin}
\end{align}
where $\bb=(\bb_1,\dots,\bb_{3m})$ is uniform over $\zo^{3m}$; it should be noted that in \Cref{eq:gourd} the randomness in the two random variables whose variation distance is being considered is only over $\bb_1,\dots,\bb_m$ and $\bp$.

It remains to show that if $m = c {n \choose d}^{1/2}$ then $(\ref{eq:pumpkin}) \leq 0.01;$ this is equivalent to showing that
\begin{equation} \label{eq:squash}
\Ex_{\overline{\bu}}\left[\dtv\left((\bb_1,\dots,\bb_{3m}),
\pbra{\mathbf{1}\cbra{\wh{\bp}\cdot \Phi(\bu^{i}) \geq 0} }_{i=1}^{3m}
 \right)
 \right] \leq 0.01,
\end{equation}
and this is what we establish in the rest of the proof. (Recall that $\boldf(x) = \mathbf{1}\cbra{\wh{\bp}\cdot\Phi(x) \geq 0}$.)

For $S \in {[n] \choose d}$ we write $e_S$ to denote the canonical basis vector in $\R^{{n \choose d}}$ that has a 1 in the $S$-th position and 0's in all other positions.
We observe that for $\widehat{\bp} \sim N(0,1)^{{n \choose d}}$, for any fixed $S \in {[n] \choose d}$ we have that 
\[
\Pr\left[\widehat{\bp} \cdot e_S\geq 0\right]=
\Pr\left[\widehat{\bp}(S)\geq 0\right]=1/2.
\]
Let $S_1,\dots,S_{{n \choose d}}$ be some fixed ordering, say lexicographic, of the ${n \choose d}$ elements of ${[n] \choose d}$.
Since the random variables $(\widehat{\bp}(S_i))_{i=1,\dots,3m}$ are independent, we have that the joint distribution of $(\bb_1,\dots,\bb_{3m})$ is identical to the joint distribution of $(\sign(\widehat{\bp} \cdot e_{S_1}),\dots,
\sign(\widehat{\bp} \cdot e_{S_{3m}}))$, namely both distributions are uniform over $\zo^{3m}.$ Thus we have that
\begin{align}
(\ref{eq:squash}) 
 &= \Ex_{\overline{\bu}}\left[\dtv
\left(
\pbra{\mathbf{1}\cbra{\wh{\bp}\cdot e_{S_i}} \geq 0}_{i=1}^{3m},
\pbra{\mathbf{1}\cbra{\wh{\bp}\cdot \Phi(\bu^{i}) \geq 0}}_{i=1}^{3m}
\right)
\right]
\label{eq:rigatoni}
 \end{align}
where as in the proof of \Cref{lemma:balanced}, we define 
\[\Phi: \R^n\to\R^{n\choose d} \qquad\text{as}\qquad \Phi(x) = \pbra{x_S}_{S\in {[n]\choose d}}.\]

In the rest of the argument, we show that the right hand side of \Cref{eq:rigatoni} is small by arguing that for most outcomes $(u^{(1)},\dots,u^{(3m)})$ of $(\bu^{(1)},\dots,\bu^{(3m)})$, the covariance structures of the two $3m$-dimensional Gaussians $(\widehat{\bp}\cdot e_{S_1},
\dots, \widehat{\bp} \cdot e_{S_{3m}})$ and $(\widehat{\bp} \cdot \Phi(u^{(1)}),\dots,
\widehat{\bp} \cdot \Phi(u^{(3m)}))$ are similar enough that it is possible to bound the variation distance between those two Gaussians, which implies a bound on the variation distance between the two vector-valued random variables on the right hand side of \Cref{eq:rigatoni}.

\begin{definition}
\label{def:good}
Fix a specific outcome $(u^{(1)},\dots,u^{(3m)}) \in (\bn)^{3m}$ of $(\bu^{(1)},\dots,\bu^{(3m)}) \sim (\bn)^{3m}$.
Let $\Sigma_2$ be the $3m \times 3m$ matrix whose $(i,j)$ entry is $\Phi'(u^{(i)}) \cdot \Phi'(u^{(j)})$, let $I_{3m}$ be the $3m \times 3m$ identity matrix, and let $A=\Sigma_2 - I_{3m}.$ We say that $(u^{(1)},\dots,u^{(3m)})$ is \emph{bad} if $\tr(A^2) \geq {\frac 1 {90000}}$; otherwise we say that $(u^{(1)},\dots,u^{(3m)})$ is \emph{good}.
\end{definition}
Intuitively, the tuple $(u^{(1)},\dots,u^{(m)})$ is bad if the $3m$ vectors $\Phi(u^{(1)}),\dots,\Phi(u^{3m})$ are not such that ``most pairs are close to orthogonal to each other'' (recall that each vector $\Phi(u^{(i)})$ is ${n \choose d}$-dimensional). As we now show, this is very unlikely to happen for $(\bu^{(1)},\dots,\bu^{(3m)}) \sim (\bn)^{3m}$.

\begin{claim} \label{claim:good}
For $m = c{n \choose d}^{1/2}$, we have
$\Pr[(\bu^{(1)},\dots,\bu^{(3m)})$ is bad$]\leq 0.005.$
\end{claim}
\begin{proof}
The idea is to apply Chebyshev's inequality to the random variable $\tr(A(\bu^{(1)},\dots,\bu^{(3m)}))$.  In more detail, we observe that $\tr(A(\bu^{(1)},\dots,\bu^{(3m)})^2)$ is a polynomial in the $3mn$ uniform random Boolean variables $\bu^{(r)}_s$.
Let's start by working out what this polynomial is (for notational ease, let's write $q=q(u)$ for this polynomial).
Since $\Phi'(u^{(i)})$ is a unit vector for all $i$ for every outcome $u^{(i)}$ of $\bu^{(i)}$, the diagonal entries of $A$ are zero. For $i \neq j$, the entry $A_{ij}$ is
\begin{equation} \label{eq:Aij}
A_{ij}={\frac {\Phi(u^{(i)}) \cdot \Phi(u^{(j)})}{{n \choose d}}}={\frac {\sum_{S \in {[n] \choose d}} u^{(i)}_S u^{(j)}_S}{{n \choose d}}},
\end{equation}
so the diagonal element $(A^2)_{ii}$ of $A^2$ is
\begin{align} 
(A^2)_{ii} &= \sum_{j=1}^{3m} A_{ij}A_{ji}= \sum_{j \in [3m] \setminus \{i\}} A_{ij}^2 \tag{using symmetry of $A$ and the fact that $A_{ii}=0$} \nonumber \\
&=\sum_{j \in [3m] \setminus \{i\}} {\frac 1 {{n \choose d}^2}} \sum_{S,T \in {{[n] \choose d}}} u^{(i)}_{S \,\triangle\, T} u^{(j)}_{S \,\triangle\, T},
\label{eq:Asquaredii}
\end{align}
and 
\begin{align} 
q(u)=\tr(A^2)&= \sum_{i=1}^{3m}  \sum_{j \in [3m] \setminus \{i\}}{\frac 1 {{n \choose d}^2}} \sum_{S,T \in {{[n] \choose d}}} u^{(i)}_{S \,\triangle\, T} u^{(j)}_{S \,\triangle\, T}.
\label{eq:traceAsquared}
\end{align}
From (\ref{eq:traceAsquared}) we see that $\tr(A^2)=q$ is a polynomial of degree at most $4d$ (since $|S \,\triangle\, T|$ is always at most $2d$), and that the constant term of $q$ is 
\[
\E[q]=\widehat{q}(\emptyset)=3m(3m-1) \cdot {\frac 1 {{n \choose d}^2}} \cdot {n \choose d}, 
\]
which is at most $9c^2$ recalling that $m=c {n \choose d}^{1/2}.$ We choose the constant $c$ so that $9c^2 \leq {\frac 1 {180000}}$, so $\E[q] \leq {\frac 1 {180000}}$; we want to show that for a random $\bu = (\bu^{(i)}_j)$, the probability that $q(\bu)$ exceeds its mean by ${\frac 1 {180000}}$ is at most a constant. To do this, let us bound the variance of $q$, i.e. the sum of squares of its non-constant coefficients.

Note that in \Cref{eq:traceAsquared}, we have that every non-constant coefficient is on a monomial of the form $u^{(i)}_U u^{(j)}_U$ for some $i \neq j$ where $U\sse[n]$ and $|U| =: k \leq 2d.$ Recall from \Cref{lem:linearization},
the number of $(S,T)$ pairs with $|S|,|T| = d$ such that $S\,\triangle\, T = U$ is at most 
\[O_d\pbra{n^{d-\ceil{k/2}}}.\]
So the coefficient of $u^{(i)}_Cu^{(j)}_C$ is
\[
\widehat{q}(i,j,C) \leq {\frac {O_d\pbra{n^{d-\ceil{k/2}}}}{{n \choose d}^2}}.
\]
As there are at most ${n\choose k}$ sets $U\sse[n]$ with $|U| = k$, we have
\begin{align*}
	\Var[q] &\leq 9m^2\sum_{k=1}^{2d} {n\choose k}\cdot \widehat{q}(i,j,C)^2\\
	&\leq 9m^2\sum_{k=1}^{2d} {n\choose k}\cdot \pbra{\frac {O_d\pbra{n^{d-\ceil{k/2}}}}{{n \choose d}^2}}^2\\
	&\leq O_d\pbra{9m^2\sum_{k=1}^{2d} \frac{n^k\cdot n^{2d-k}}{n^{4d}}}\\
	&\leq O_d\pbra{\frac{1}{n^d}},
\end{align*}
where the final inequality relies on our choice of $m := c {n \choose d}^{1/2}$. Putting everything together, we have that
\[\E[q]\leq {\frac 1 {180000}} \qquad\text{and}\qquad \sqrt{\Var[q]} \leq O_d\pbra{\frac{1}{n^{d/2}}}.\]
By Chebyshev's inequality, we get
\[
\Pr\sbra{q \geq {\frac 1 {90000}}} \leq O\pbra{\frac{1}{n^{d}}}\]
which is at most $0.005$ for sufficiently large $n$. In particular, this implies that for $n$ large enough, we have $\Pr[(\bu^{(1)},\dots,\bu^{(3m)})$ is bad$]\leq 0.005$, completing the proof.
\ignore
{
Looking at \Cref{eq:traceAsquared} we see that every non-constant coefficient is on a monomial of the form $u^{(i)}_C u^{(j)}_C$ for some $i \neq j$ and some $|C| \leq 2d.$  {For a fixed pair $i \neq j$ and a fixed $C$ with $|C|=i > 0$, there are ${n-i \choose d-i}$ pairs $(S,T)$ such that $S \,\triangle\, T = C$ (because $S$ and $T$ must agree on all $d-i$ elements outside $C$, and these are chosen from $n-i$ possible elements).} So the coefficient of $u^{(i)}_Cu^{(j)}_C$ is
\[
\widehat{q}(i,j,C) = {\frac {{n-i \choose d-i}}{{n \choose d}^2}} \leq {\frac {(d/n)^i}{{n \choose d}}}
\]
and since there are at most $n^i$ sets $C$ with $|C|=i$, we have
\begin{align}
\Var[q] &\leq 9m^2 \sum_{i=1}^{2d} n^i {\frac {(d/n)^{2i}}{{n \choose d}^2}} = 9c^2 \sum_{i=1}^{2d} {\frac {d^{2i}/n^i}{{n \choose d}}} \\
&\leq d^{5d} \cdot {\frac 1 {n {n \choose d}}}.
\end{align}

So we have $\E[q]\leq {\frac 1 {180000}}$ and $\sqrt{\Var[q]} \leq d^{5d/2} {\frac 1 {\sqrt{n {n \choose d}}}}$, so by Chebyshev's inequality we have $\Pr[q \geq {\frac 1 {90000}}] \leq O({\frac {d^{5d}}{n{n \choose d}}})$, which is at most $0.005$ for sufficiently large $n$.
 So $\Pr[(\bu^{(1)},\dots,\bu^{(3m)})$ is bad$]\leq 0.005.$
 }
\end{proof}

Hence to show that $(\ref{eq:rigatoni}) \leq 0.01$ (and finish the proof of \Cref{thm:lb-ptf}), it suffices to prove the following:

\begin{claim} 
\label{claim:combray}
Fix a good $(u^{(1)},\dots,u^{(3m)}) \in (\bn)^{3m}$. 
Then for $\widehat{\bp} \sim N(0,1)^{n \choose d}$ we have that
\begin{equation} \label{eq:persuasion}
\dtv
\left(
\pbra{\mathbf{1}\cbra{\wh{\bp}\cdot e_{S_i}} \geq 0}_{i=1}^{3m},
\pbra{\mathbf{1}\cbra{\wh{\bp}\cdot \Phi(\bu^{i}) \geq 0} }_{i=1}^{3m}
\right)\leq 0.005.
\end{equation}
\end{claim}
To prove \Cref{claim:combray} we first observe that (\ref{eq:persuasion}) is equivalent to
\begin{equation}\label{eq:pride-and-prejudice}
\dtv
\left(
\pbra{\mathbf{1}\cbra{\wh{\bp}\cdot e_{S_i}} \geq 0}_{i=1}^{3m},
\pbra{\mathbf{1}\cbra{\wh{\bp}\cdot \Phi'(u^{i}) \geq 0 } }_{i=1}^{3m}
\right) \leq 0.005
\end{equation} 
where (as in the proof of \Cref{lemma:balanced}) we define $\Phi'(u^{(i)}) := \Phi(u^{(i)})/\|\Phi(u^{(i)})\| =
\Phi(u^{(i)})/\sqrt{{n \choose d}}$  to be $\Phi(u^{(i)})$ normalized to unit length. 
We now apply \Cref{thm:DMR20} where we take $k=3m$, $\Sigma_1$ to be the identity matrix $I_{3m},$ and $\Sigma_2$ to be the $3m \times 3m$ matrix whose $(i,j)$ entry is $\Phi'(u^{(i)}) \cdot \Phi'(u^{(j)})$, so a draw from $N(0,\Sigma_1)$ is distributed as $(\widehat{\bp} \cdot e_1,\dots,\widehat{\bp} \cdot e_{3m})$ and a draw from $N(0,\Sigma_2)$ is distributed as
$(\widehat{\bp} \cdot \Phi'(u^{(1)}),\dots,\widehat{\bp} \cdot \Phi'(u^{(3m)})).$  Applying the data processing inequality for total variation distance (see e.g. Proposition~B.1 of \cite{DDOST13}), it follows that the LHS of (\ref{eq:pride-and-prejudice}) is at most $\dtv(N(0^{3m},\Sigma_1),N(0^{3m},\Sigma_2))$; we proceed to upper bound the RHS of \Cref{thm:DMR20}.

Let $A$ denote the matrix $\Sigma_1^{-1}\Sigma_2 - I_{3m}$, so $\lambda_1^2,\dots,\lambda_{3m}^2$ are the eigenvalues of $A^2$  and we have $\sqrt{\lambda_1^2 + \cdots + \lambda_{3m}^2} = \sqrt{\tr(A^2))}.$ Since $A$ is good we have that $\tr(A^2) \leq {\frac 1 {90000}}$, which gives that ${\frac 3 2} \cdot \min\left\{1,\sqrt{\lambda_1^2 + \cdots + \lambda_m^2}\right\} \leq 0.05$ as required, and the proofs of \Cref{claim:combray}, \Cref{eq:squash}, and \Cref{thm:lb-ptf} are complete.

\section*{Acknowledgements}
\label{sec:ty}

A.D. is supported by NSF grants CCF-1910534, CCF-1926872, and CCF-2045128. H.L. is supported by NSF grants CCF-1910534, CCF-1934876, and CCF-2008305. S.N. is supported by NSF grants  IIS-1838154, CCF-2106429, CCF-2211238, CCF-1763970, and CCF-2107187. Part of this work was completed while S.N. was participating in the program on ``Meta Complexity'' at the Simons Institute for the Theory of Computing. R.A.S. is supported by NSF grants  IIS-1838154, CCF-2106429, and CCF-2211238. The authors would like to thank the anonymous STOC reviewers for helpful comments. 

\bibliography{allrefs}
\bibliographystyle{alpha}

\appendix

\section{Testing Truncation via VC Dimension}
\label{appendix:baseline-distinguisher}

In this section, we show that $O_{d}(n^{d}/\eps^2)$ samples are sufficient for distinguishing an unknown distribution $\calP$ over $\R^n$ from $\calP|_{f^{-1}(1)}$ where $f:\R^n\to\zo$ is a degree-$d$ PTF with $\vol(f) \leq 1-\eps$. This follows immediately from the following more general result which allows us to test truncation by an unknown function belonging to a concept class $\calC$ with sample complexity linear in the \emph{VC-dimension} of the class $\calC$ (denoted $\vcdim(\calC)$; we refer the reader to \cite{KearnsVazirani:94} for background on  VC-dimension) {and bounds on the VC-dimension of degree-$d$ PTFs over $\R^n$ (Corollary~9 of~\cite{anthony1995classification})}.  As mentioned in the introduction, the algorithm used to prove \Cref{prop:vc} is not time-efficient. 


\begin{proposition}  \label{prop:vc}
  Let ${\calP}$ be an arbitrary probability distribution over $\R^n$.
  Let $\calC$ be a class of $0/1$-valued functions on $\R^n$ such that 
  each $f\in\calC$ satisfies 
  \[\vol(f) \leq 1-\eps.\]
There exists an algorithm with the following properties:  The algorithm is given access to i.i.d. samples from an unknown distribution $\calD$ with \[\calD \in \cbra{\calP} \cup \cbra{\calP|_{f^{-1}(1)} : f\in\calC}.\] It draws $\Theta\pbra{\vcdim(\calC)/\eps^2}$ samples from $\calD$, and has the following performance guarantee:
  \begin{enumerate}
  	\item If $\calD = \calP$, then the algorithm outputs ``un-truncated'' with probability at least $2/3$; and 
  	\item If $\calD = \calP|_{f^{-1}(1)}$ for some $f\in\calC$, then the algorithm outputs ``truncated'' with probability 1.
  \end{enumerate}  
\end{proposition}


\begin{proof}
Consider the following simple algorithm: 
\begin{enumerate}
	\item Draw $T := \Theta\pbra{\frac{\vcdim(\calC)}{\epsilon^2}}$ samples $\{\x{1},\ldots, \x{T}\}$ from $\calD$. 
	\item If there exists a function $f \in \calC$ such that $f(\x{i}) =1$ for all $i \in [T]$, output ``truncated.'' Otherwise, output ``untruncated.''
\end{enumerate}


If $\calD = \calP|_{f^{-1}(1)}$ for some $f\in\calC$, then the algorithm will clearly always output ``truncated.'' It remains to show that when $\calD = \calP$, the algorithm outputs ``truncated'' with probability at most $1/3$. As $\vol(f) \leq 1-\eps$ for all $f\in \calC$, we have that for any fixed $f\in\calC$ the following holds by linearity of expectation: 
\[\Ex_{\x{1},\dots,\x{T}\sim\calP}\sbra{\frac{\#\cbra{i : f(\x{i}) = 1}}{T}} \leq 1-\eps.\]
By a standard application of the VC inequality~\cite{VC71} (see e.g.~Theorem 4.3 of \cite{AnthonyBartlett}) and recalling our choice of $T$, we get that 
\[\Prx_{\x{1},\dots,\x{T}\sim\calP}\sbra{\frac{1}{T}\cdot\sup_{f\in\calC}{\#\cbra{i : f(\x{i}) = 1}} \geq 1-\frac{\eps}{4}} \leq 0.01.\]
In particular, this implies that with probability at least 99/100, no function $f\in\calC$ is consistent with $f(\x{i})$ for $i\in[T]$ when $\calD = \calP$, completing the proof.
\end{proof}

\section{An $\Omega(n^{d/4})$-Sample Lower Bound for Testing Low-Volume PTF Truncation}
\label{appendix:weak-lower-bound}

In this section, we give a simple lower bound of $\Omega(n^{d/4})$-samples to test PTF truncation. In more detail, we show that at least $\Omega(n^{d/4})$-samples are required to distinguish the uniform distribution on the Boolean hypercube $\bn$ from the uniform distribution on $\bn$ truncated by a \emph{low-volume} degree-$d$ PTF; the volume of the truncating PTF is in fact be only $\Theta(n^{d/2}/2^{n})$.

Towards this, we first recall a useful lemma due to Aspnes et al.~\cite{ABF+:94} (whose proof only uses basic linear algebra):
\begin{lemma}[\cite{ABF+:94}, Lemma~2.1]\label{lem:ptf}
  Let $S$ be a set of points in $\{-1,1\}^n$ such that
  $|S| < \sum_{i=0}^{d/2} \binom{n}{i}$ for some even number $0 < d \leq n$. Then there exists a degree-$d$ PTF $f$ whose satisfying assignments are precisely the points in $S$, i.e.~$f(x) = 1$ for all $x \in S$ and $f(x) = 0$ for all $\bx\notin S$.
\end{lemma}

We now present our lower bound and its proof.

\begin{proposition}
  Let ${\cal U}$ be the uniform distribution over $\bn$.
  Fix any constant $d$, and let $\Pptf^{\ast}$ be the class of Boolean functions on $\bn$
  such that each $f\in\Pptf^{\ast}$ is a degree-$d$ PTF $f: \bn\to \{0,1\}$ with 
  \[\vol(f) = {\frac {|f^{-1}(1)|}{2^n}}= \Theta(n^{d/2}/2^{n}).\]
  Let $A$ be any algorithm with the following property: Given sample access to an unknown distribution $\calD$ with 
  \[\calD \in \{\calU\} \cup \cbra{\calU|_{f^{-1}(1)} : f\in\Pptf^{\ast}},\] $A$ has the following performance guarantee:
  \begin{enumerate}
  	\item If $\calD = \calU$, then $A$ outputs ``un-truncated'' with probability $2/3$; and 
  	\item If $\calD = \calU|_{f^{-1}(1)}$ for some $f\in\Pptf^\ast$, then $A$ outputs ``truncated'' with probability $2/3$.
  \end{enumerate}
  Then $A$ must use at least $\Omega(n^{d/4})$ samples from $\calD$.
\end{proposition}

\begin{proof}
  Consider choosing a degree-$d$ PTF over $\bn$ as follows:
  \begin{enumerate}
  	\item First, draw a collection $\bS$ of $N_{d/2} := \sum_{i=0}^{d/2} {n \choose i} - 1 = \Theta(n^{d/2})$ random points from $\bn$. 
  	\item Let $\boldf$ be a degree-$d$ PTF such that   
  	\[\boldf(x) = \begin{cases}
  		1 & x\in \bS\\
  		0 & x\notin\bS 
  	\end{cases}.\]
  	Note that the existence of such an $\boldf$ is guaranteed by \Cref{lem:ptf}.
  \end{enumerate}
  
  We will now show that no algorithm $A$ that draws $m = cn^{d/4}$ samples, for a sufficiently small absolute constant $c>0$, can distinguish between $\calU$ and $\calU|_{\boldf^{-1}(1)}$ with probability $2/3$, from which the result is immediate. 
When $\calD = \calU$, then note that the $m$ samples are distribution independently and uniformly over $\bn$. On the other hand, when $\calD = \calU|_{\boldf^{-1}(1)}$, then (i) by a standard application of the birthday paradox, since $|\boldf^{-1}(1)| \geq \Omega(n^{d/2})$, the $m$ samples are all distinct except with an arbitrarily small constant probability (depending on $c$); and (ii) conditioned on all the samples being distinct, they are also uniformly distributed over $\bn$ as the support of $\boldf$ (namely $\bS$) was chosen at random. 

It follows from (i) and (ii) that the total variation distance between $\calU$ and $\calU|_{\boldf^{-1}(1)}$ is at most an arbitrarily small constant, as needed to be shown.
\end{proof}

\end{document}